%% file: main.tex
\documentclass[sigconf, authorversion]{acmart} 
\makeatletter
\def\@ACM@checkaffil{
    \if@ACM@instpresent\else
    \ClassWarningNoLine{\@classname}{No institution present for an affiliation}%
    \fi
    \if@ACM@citypresent\else
    \ClassWarningNoLine{\@classname}{No city present for an affiliation}%
    \fi
    \if@ACM@countrypresent\else
        \ClassWarningNoLine{\@classname}{No country present for an affiliation}%
    \fi
}
\makeatother

\usepackage[ruled, linesnumbered]{algorithm2e}
\usepackage{amsthm}
\usepackage{amsmath}
\usepackage{amsfonts}
\usepackage{bold-extra}
\usepackage{booktabs}
\usepackage{caption}
\usepackage{diagbox}
\usepackage{enumerate}
\usepackage{enumitem}
\usepackage{float}
\usepackage{graphicx}
\usepackage{multirow}
\usepackage{subfigure}
\newtheorem{lemma}{Lemma}
\newtheorem{theorem}{Theorem}
\newtheorem{proposition}{Proposition}

\newtheorem{property}{Property}

\AtBeginDocument{%
  \providecommand\BibTeX{{%
    \normalfont B\kern-0.5em{\scshape i\kern-0.25em b}\kern-0.8em\TeX}}}

\copyrightyear{2023}
\acmYear{2023}
\setcopyright{acmlicensed}
\acmConference[KDD '23] {Proceedings of the 29th ACM SIGKDD Conference on Knowledge Discovery and Data Mining}{August 6--10, 2023}{Long Beach, CA, USA.}
\acmBooktitle{Proceedings of the 29th ACM SIGKDD Conference on Knowledge Discovery and Data Mining (KDD '23), August 6--10, 2023, Long Beach, CA, USA}
\acmPrice{15.00}
\acmDOI{10.1145/3580305.3599250}
\acmISBN{979-8-4007-0103-0/23/08}
\begin{document}

\title{Accelerating Dynamic Network Embedding with Billions of Parameter Updates to Milliseconds}

\author{Haoran Deng}
\affiliation{%
  \institution{Zhejiang University}
}
\email{denghaoran@zju.edu.cn}
\orcid{0009-0008-5706-268X}

\author{Yang Yang}
\authornote{Corresponding author.}
\affiliation{%
  \institution{Zhejiang University}
}
\email{yangya@zju.edu.cn}
\orcid{0000-0002-5058-4417}

\author{Jiahe Li}
\affiliation{%
  \institution{Zhejiang University}
}
\email{jiahe.20@intl.zju.edu.cn}
\orcid{0009-0003-5092-1326}

\author{Haoyang Cai}
\affiliation{%
  \institution{Carnegie Mellon University}
}
\email{hcai2@andrew.cmu.edu}
\orcid{0009-0009-1635-6910}

\author{Shiliang Pu}
\affiliation{%
  \institution{Hikvision Research Institute}
}
\email{pushiliang.hri@hikvision.com}
\orcid{0000-0001-5269-7821}

\author{Weihao Jiang}
\affiliation{%
  \institution{Hikvision Research Institute}
}
\email{jiangweihao5@hikvision.com}
\orcid{0000-0003-3482-8538}

\renewcommand{\shortauthors}{Haoran Deng, et al.}

\begin{abstract}
  \input{00.tex}
\end{abstract}




\maketitle

\input{01Introduction.tex}

\input{02Preliminaries.tex}
\input{03RelatedWork.tex}

\input{04Methodology.tex}
\input{05Experiment.tex}
\input{06Conclusion.tex}

\bibliographystyle{ACM-Reference-Format}
\balance
\bibliography{sample-base}

\appendix
\input{07Appendix.tex}

\end{document}

%% file: 00.tex
Network embedding, a graph representation learning method illustrating network topology by mapping nodes into lower-dimension vectors, is challenging to accommodate the ever-changing dynamic graphs in practice.
Existing research is mainly based on node-by-node embedding modifications, which falls into the dilemma of efficient calculation and accuracy.
Observing that the embedding dimensions are usually much smaller than the number of nodes, we break this dilemma with a novel dynamic network embedding paradigm that rotates and scales the axes of embedding space instead of a node-by-node update.
Specifically, we propose the Dynamic Adjacency Matrix Factorization (DAMF\footnote{The code is available at https://github.com/zjunet/DAMF}) algorithm, which achieves an efficient and accurate dynamic network embedding by rotating and scaling the coordinate system where the network embedding resides with no more than the number of edge modifications changes of node embeddings.
Moreover, a dynamic Personalized PageRank is applied to the obtained network embeddings to enhance node embeddings and capture higher-order neighbor information dynamically.
Experiments of node classification, link prediction, and graph reconstruction on different-sized dynamic graphs suggest that DAMF advances dynamic network embedding.
Further, we unprecedentedly expand dynamic network embedding experiments to billion-edge graphs, where DAMF updates billion-level parameters in less than 10ms.

%% file: 01Introduction.tex
\section{Introduction}


Network embedding is an advanced graph representation learning method that maps each node in a graph to a vector in a low-dimensional space while preserving the graph's topological information~\cite{Arsov2019NetworkEA}. 
It is a versatile approach with successful practical applications in a wide range of fields, such as recommendation systems and bioinformatics~\cite{cui2018survey, shi2018heterogeneous, wen2018network,Choudhary2021ASO, nelson2019embed, su2020network}.

    
In practice, however, changes in graph structure are frequent and inevitable. 
For instance, in a social network, the nodes representing the new members are added in connection to the original network over time, forming new edges. 
On such dynamic graphs, the network embeddings should be updated with the transformation of the graphs to empower the model to capture crucial insights, such as which groups are more active or which members are more likely to be influencers. Moreover, the temporal evolution contributes to the portrait of the new members.
Unfortunately, the frequent evolving patterns and the extensive network scale require enormous time and space to retrain the network embedding to model the dynamic network effectively.





Recent research has encountered a dilemma in efficiency and effectiveness. 
That is, precisely modifying the network embeddings leads to excessive inefficiencies.
Some methods~\cite{chen2015fast, zhang2018timers, zhu2018high} choose to adjust the embedding of nearly every node (also known as global updates), resulting in high time complexity.
On the other hand, 
some works~\cite{liu2019real, du2018dynamic, hou2020glodyne, mahdavi2018dynnode2vec} make trade-offs by updating only the more affected nodes in the graph (also known as local updates) to ensure good efficiency, but constantly bringing errors and, consequently, leading to performance degradation.
Moreover, existing applications like short video instant recommendations require high instantaneity, which indicates
it is inappropriate to use delayed retraining methods, like gathering multiple modifications for a single update.


We break the dilemma by rotating and scaling the coordinate axes of the embedding space instead of updating individually from the nodes' perspective.
By calculating the space projection matrices, the newly added edge is captured while retaining the semantics of the embedding from the previous step. 
By means of an embedding space projection matrix and a small number of modifications of the node embedding, each node vector embedding in the graph will be relocated at the updated position in the embedding space.
This approach is efficient since the number of coordinate axes (i.e., the dimension of the embedding) is significantly fewer than the number of nodes while retaining a high accuracy level compared to local updates.


In light of the preceding ideas, we propose the Dynamic Adjacency Matrix Factorization (DAMF) algorithm.
The corresponding space projection matrix is solved based on \emph{Zha-Simon's t-SVD update formula}~\cite{zha1999updating}, with additional modification at most $\Delta m$ nodes' embedding, where $\Delta m$ is the number of edge changes in the graph. 
Further, inspired by the great success of the application of Personalized PageRank (PPR) to static network embeddings~\cite{yang2020nrp, tsitsulin2018verse, yin2019strap} and graph neural networks~\cite{gasteiger2018predict}, we use a dynamic PPR to enhance the network embedding in order to capture high-order neighbors' information.

With the above design and ideas, DAMF also provides theoretical guarantees of effectiveness and efficiency. For effectiveness specifically, let $\widetilde{\mathbf{A}}$ be the low-rank approximation of the adjacency matrix represented by current network embedding, the unenhanced DAMF achieves the matrix factorization of the updated $\widetilde{\mathbf{A}}$ with node or edge change (Theorem \ref{theorem:node} \& \ref{theorem:edge} ). Moreover, the dynamic embedding enhancement converges into accurate PPR propagation. For efficiency, the time complexity of the DAMF algorithm is proved to be $O(\Delta m)$ when hyperparameters are considered constants.



In addition, we are the first to extend dynamic network embedding experiments to billion-edge graphs, which is a breakthrough in the scale of massive graphs.
We conduct our experiment on Twitter, a real-world graph dataset with 41 million nodes and \textbf{1.4 billion} edges, and map each node to a $128$-dimensional vector with the number of learnable parameters more than 10 times that the BERT-Large's~\cite{kenton2019bert}.
We add nodes to the graph and updated the network embeddings individually, starting from 1000 nodes. The total updating time of DAMF is 110 hours, illustrating that DAMF achieves billion-level parameter updates in less than 10ms.
In addition, we have conducted experiments on node classification, link prediction, and graph reconstruction on graphs of different sizes. The experimental results show that DAMF reaches a new state of the art in dynamic network embedding.

To summarize, we make the following contributions:
\begin{itemize}[leftmargin=*, topsep=2pt, itemsep=1pt]
\item We present the DAMF algorithm, a novel dynamic network embedding method based on embedding space projections, and augment the embedding with dynamic PPR to capture higher-order neighbourhood information.

\item We theoretically illustrate the efficiency and effectiveness of DAMF.

\item To the best of our knowledge, we are the first to extend dynamic network embedding to a dataset with billions of edges. Experimental results show that our proposed methods use only an average of $10ms$ to complete global parameter updates. We also conduct experiments on five other actual network data, and the results show that the proposed methods reach a new state of the art in dynamic network embedding.
\end{itemize}

%% file: 02Preliminaries.tex
\section{preliminaries}
\begin{table}[htbp]
\newcommand{\tabincell}[2]{\begin{tabular}{@{}#1@{}}#2\end{tabular}}
    \caption{Notations}
    \label{tab:notaions}
    \begin{tabular}{p{1.1cm} | p{6.5cm}}
        \toprule
        Notation & Description\\
        \midrule
        $\mathcal{G}(\mathcal{V},\mathcal{E})$ & the graph with node set $\mathcal{V}$ and edge set $\mathcal{E}$ \\
        $n, m$ & the number of nodes and the number of edges \\
        $\Delta m$ & the number of edges change in graph \\
        $deg(u)$ & the degree (in-degree or out-degree) of node $u$ \\
        $\mathbf{A}$, $\mathbf{D}$  & the adjacency matrix and degree matrix \\
        $\mathbf{B}$, $\mathbf{C}$ & the low-rank representation of the updated matrix. \\
        $\mathbf{X}$, $\mathbf{Y}$ & the context embedding and content embedding \\
        $\mathbf{Z}$ & the enhanced context embedding \\
        $\mathbf{F}, \mathbf{G}$ & the space projection matrix for $\mathbf{X}$ and $\mathbf{Y}$ \\
        $\Delta \mathbf{X}, \Delta \mathbf{Y}$ & the node vectors that are directly modified \\
        $d$ & the dimension of the embedding space \\
        $\alpha$ & the damping factor in \textit{Personalized PageRank} \\
        $\epsilon$ & the error tolerance in \textit{Personalized PageRank} \\
        \bottomrule
\end{tabular}
\end{table}

Consider a graph $\mathcal{G}=(\mathcal{V}, \mathcal{E})$, where $\mathcal{V}$ denotes the node set of $n$ nodes and $\mathcal{E}$ denotes the edge set of $m$ edges. 
\emph{Dynamic graph scenario} is a sequential update events $\{Event_1, ..., Event_T \}$ with an initial graph $\mathcal{G}_0$. Each event is one of either node change or edge change.
Let $\mathbf{A} \in \mathbb{R}^{n\times n}$ be the adjacency matrix of $\mathcal{G}$, and $\mathbf{D}=diag\{deg[1],
deg[2], ... , deg[n]\}$ be the diagonal degree matrix where $\deg[i]=\sum_j \mathbf{A}[i, j]$ is the out-degree of $i$-th node. 
Network embedding aims at mapping each node in graph $\mathcal{G}$ to one or two low-dimensional vectors, which capture each node's structural information in the graph. In this paper, for a given dimension $d\ll n$, the $i$-th node in graph $\mathcal{G}$ is mapped to two vectors $\mathbf{X}[i], \mathbf{Y}[i] \in \mathbb{R}^{\frac{d}{2}}$ with equal dimension, which capture the structural information. Dynamic Network Embedding updates the result of the embedding for the $t$-th event, and the snapshot of the updated embedding is noted as $\mathbf{X}_t$, $\mathbf{Y}_t$.

In this paper, matrices are denoted in bold uppercase letters. Let $\mathbf{M}$ be arbitrary matrix, $\mathbf{M}[i]$ is the $i$-th row vector of $\mathbf{M}$, $\mathbf{M}[:,j]$ is the $j$-th column vector of $\mathbf{M}$, and $\mathbf{M}[i,j]$ is the element on the $i$-th row $j$-th column of $\mathbf{M}$. In addition, we use $\mathbf{M}[:l]$ to denote the submatrix consisting of the first $l$ rows of $\mathbf{M}$, and use $\mathbf{M}[l:]$ to denote the submatrix consisting of the remaining part of $\mathbf{M}$.

%% file: 03RelatedWork.tex
\section{Related Work}
\subsection{Network Embedding}
Network embedding aims to reflect the structural information of nodes in a graph by mapping each node into a low-dimensional vector~\cite{hamilton2017representation}.
Compared to another popular graph learning method, Graph Neural Networks (GNNs), it has no requirement for any features or attributes of nodes.

Studies on network embedding can simply be classified into gradient-based methods and matrix factorization-based methods.
Gradient-based methods like DeepWalk~\cite{perozzi2014deepwalk}, node2vec~\cite{grover2016node2vec}, and LINE~\cite{tang2015line} learn a skip-gram with negative sampling from random-walk on the graph, while GraphGAN ~\cite{wang2018graphgan}, GA~\cite{abu2018watch}, DNGR~\cite{cao2016deep} use deep learning method to learn network embedding.
For matrix factorization-based methods, AROPE~\cite{zhang2018arbitrary}, NetMF~\cite{qiu2018netmf}, NetSMF~\cite{qiu2019netsmf},  and LightNE~\cite{qiu2021lightne} construct a matrix that reflects the properties of the graph and factorizes it to obtain the network embedding, while ProNE~\cite{zhang2019prone} and NRP ~\cite{yang2020nrp} obtain network embedding by propagates the embedding on the graph after the factorization.

Among the network embedding methods, those who use PageRank~\cite{page1999pagerank}, such as VERSE~\cite{tsitsulin2018verse}, STRAP~\cite{yin2019scalable}, and NRP~\cite{yang2020nrp}, have achieved an ideal result since they better capture the information of high-order neighbors.

\subsection{Dynamic Network Embedding}
Numerous graphs in the industry are dynamic with frequent node or edge modifications, which leads to the emergence of network embedding methods on dynamic graphs. Methods for dynamic graph embedding can be categorized as node-selection-based, matrix-factorization-based, and others.

\textbf{Node-selection-based methods.}
Node-selection-based methods choose to update only a limited number of embeddings node-by-node, resulting in poor performance.
DNE~\cite{du2018dynamic} updates the embedding of nodes by adapting the skip-gram to the dynamic graph;
Dynnode2vec~\cite{mahdavi2018dynnode2vec} fine-tunes the previous result with the newly sampled data;
GloDyNE~\cite{hou2020glodyne} improves the node selection strategy to guarantee a global topological relationship;
LocalAction~\cite{liu2019real} achieves an efficient and dynamic network but with poor performance by randomly selecting a small number of neighbors around the updated nodes and modifying their embedding.

\textbf{Matrix-Factorization-based methods.}
Matrix-factorization-based methods prefer global updates that adjust the embedding of almost every node, which leads to high time complexity.
TRIP~\cite{chen2015fast} efficiently updates the top eigen-pairs of the graph, but leads to a significant accumulation of errors; TIMERS~\cite{zhang2018timers} uses recalculation to mitigate the errors caused by TRIP;
RandNE~\cite{zhang2018billion} uses a random projection method to update the factorization of the adjacency matrix; DHEP ~\cite{zhu2018high} modifies the most affected eigenvectors using the matrix perturbation theory. Unfortunately, regardless of the amount of changes, the time complexity of these methods except RandNE for a graph with $n$ nodes is at least $O(n)$.

\textbf{Other methods.} DynGEM ~\cite{goyal2018dyngem} and NetWalk~\cite{yu2018netwalk} use an auto-encoder to continuously train the model with parameters inherited from the last time step with a regularization term. 
DepthLGP~\cite{ma2018depthlgp} considers adding nodes to the dynamic graph as if they were out-of-sample and interpolating their embedding. However, the above methods are difficult to adapt to frequently changed graphs or cold start scenarios.

As a distinction, the dynamic network embedding in this paper tries to adapt embedding updates to dynamically changing graphs instead of mining the graph for temporal information(e.g., DynamicTraid~\cite{zhou2018dynamictrad}, CTDNE~\cite{nguyen2018ctdne}, DTINE~\cite{gong2020ctine}, tNE\cite{singer2019tne}), and the nodes have no attributes or features(e.g., EvolveGCN~\cite{pareja2020evolvegcn}, DANE~\cite{li2017dane}).

%% file: 04Methodology.tex
\section{Methodology}

In this section, we develop the Dynamic Adjacency Matrix Factorization (DAMF) algorithm.
Figure \ref{fig:DAMF} shows an overview of the proposed DAMF algorithm.
In what follows, Section \ref{sec:project} introduces the concept of dynamic network embeddings based on space projections.
Section \ref{sec:NodeChange} and Section \ref{sec:EdgeChange} demonstrate how to modify the space projection matrix and the node embedding for node and edge changes, respectively. Section \ref{sec:pagerank} presents an enhanced approach for dynamic network embedding using dynamic Personalized PageRank. Section \ref{sec:damf} gives the detailed steps of the DAMF. Section \ref{sec:complexity} analyzes the time and space complexity of DAMF. 

\begin{figure*}
    \centering
    \includegraphics[width=6in]{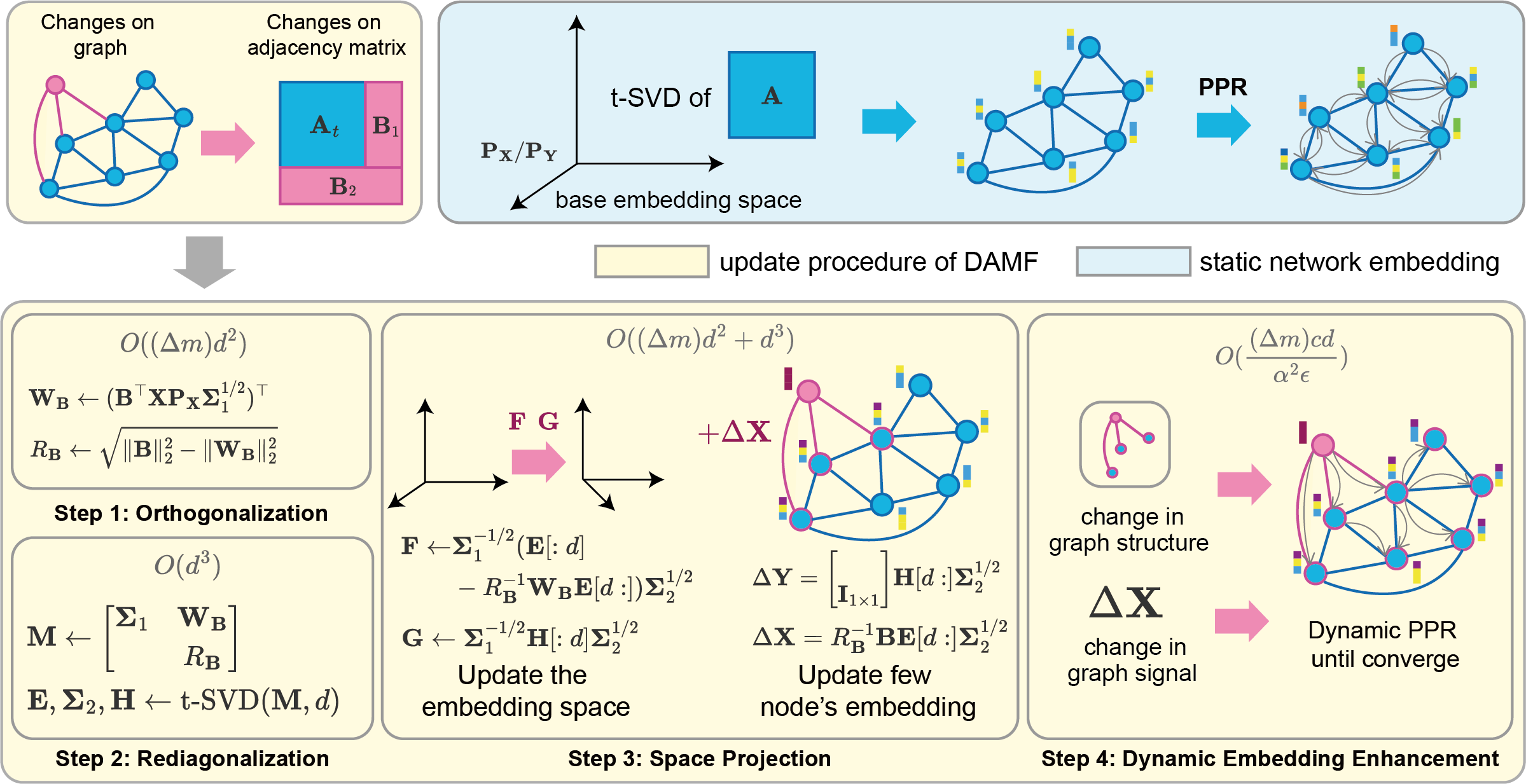}
    \caption{Overview of DAMF}
    \label{fig:DAMF}
\end{figure*}

\input{Methodology/1SpaceProjection.tex}
\input{Methodology/2NodeChange.tex}

\input{Methodology/3EdgeChange.tex}
\input{Methodology/4PageRank}

\input{Methodology/5ProposedMethod.tex}
\input{Methodology/6ComplexityAnalysis.tex}

%% file: Methodology/1SpaceProjection.tex
\subsection{Dynamic Embedding via Space Projection}
\label{sec:project}
The truncated singular value decomposition(t-SVD) of the adjacency matrix of a graph provides an ideal network embedding. For a graph $\mathcal{G}$ with adjacency matrix $\mathbf{A}\in \mathbb{R}^{n \times n}$, this method provides a network embedding $\mathbf{X} ,\mathbf{Y} \in \mathbb{R}^{n \times d}$ with dimension $d$ by
\begin{equation}
\label{eq:ne}
    \mathbf{U},\mathbf{\Sigma}, \mathbf{V}\gets \texttt{\textbf{t-SVD}}(\mathbf{A},d), \quad
    \mathbf{X}\gets \mathbf{U}\sqrt{\mathbf{\Sigma}}, \quad
    \mathbf{Y}\gets \mathbf{V}\sqrt{\mathbf{\Sigma}}, \quad
\end{equation}
In the scenario of a dynamic graph, the adjacency matrix changes over time, causing its t-SVD to change as well. 
However, the instantaneous recalculation of t-SVD is time-consuming, especially for large-scale adjacency matrices.

To cope with this problem, we propose a novel paradigm for updating network embedding by rotating and scaling (i.e., space projection) the space coordinate axes of embedding, with only a small number of nodes to update additionally.
Specifically, we take the network embedding $\mathbf{X}, \mathbf{Y}$ at time $t-1$, rotate and scale its coordinate axis by the space projection matrices $\mathbf{F}, \mathbf{G} \in \mathbb{R}^{d\times d}$, then add $\Delta X$ and $\Delta Y$ respectively to get the updated network embeddings $\mathbf{X}_t$ and $\mathbf{Y}_t$ at time $t$ in the new coordinate system by
\begin{equation}
\label{eq:raw}
    \mathbf{X}_t \gets  \mathbf{X}_{t-1} \cdot \mathbf{F} + \Delta \mathbf{X}_{t}, \quad
    \mathbf{Y}_t \gets  \mathbf{Y}_{t-1} \cdot \mathbf{G} + \Delta \mathbf{Y}_t
\end{equation}
where the \textbf{number of non-zero rows} in $\Delta \mathbf{X}_{t}$ and $\Delta \mathbf{Y}_{t}$ is the number of nodes that need to be modified for embedding.

Nevertheless, the complexity of computing $\mathbf{X}_{t-1} \cdot \mathbf{F}$ and $\mathbf{Y}_{t-1} \cdot \mathbf{G}$ is very high. 
To address this issue, we map all graph modifications to a base space that stores network embeddings at any given timestamp. By matrix multiplication, successive space projections can be merged. To improve efficiency, we apply a “lazy” query optimization that keeps the accumulated modifications until a  query of new embedding appears.

Specifically, $\mathbf{X}_0$ and $\mathbf{Y}_0$ are the initialized network embeddings of $\mathbf{X_b}, \mathbf{Y_b} \in \mathbb{R}^{n \times d}$ at timestamp $0$, while the space projection matrix $\mathbf{P_X}, \mathbf{P_Y} \in \mathbb{R}^{d \times d}$ are initialized with the identity matrix. The operation in Eq. (\ref{eq:raw}) can be transformed into
\begin{equation}
\label{eq:update}
\begin{aligned}
    \mathbf{P}_{\mathbf{X},t} \gets \mathbf{P}_{\mathbf{X}, t-1} \cdot \mathbf{F}, \quad
    \mathbf{X}_{b,t} \gets \mathbf{X}_{b, t-1} + \Delta \mathbf{X}_{t} \cdot \mathbf{P}_{\mathbf{X}, t}^{-1} \\
    \mathbf{P}_{\mathbf{Y},t} \gets \mathbf{P}_{\mathbf{Y}, t-1} \cdot \mathbf{G}, \quad
    \mathbf{Y}_{b,t} \gets \mathbf{Y}_{b, t-1} + \Delta \mathbf{Y}_{t} \cdot \mathbf{P}_{\mathbf{Y}, t}^{-1}
\end{aligned}
\end{equation}

Then, the space projection matrix and node embedding will be modified for each change to the graph.

%% file: Methodology/2NodeChange.tex
\subsection{Node Change}
\label{sec:NodeChange}

Adding nodes to a graph is equivalent to adding an equal number of rows and columns to its adjacency matrix.
Without losing generality, we consider the newly added node is in the last column and row of adjacency matrix, and we treat the column and row addition as two individual steps, that is, 

\begin{equation}
\label{eq:node_B}
    \mathbf{A}^\prime \gets [\mathbf{A}_{t-1}^{}, \mathbf{B}_1^{}], \quad
    \mathbf{A}_t^{} \gets [\mathbf{A}^{\prime \top}, \mathbf{B}_2^{}]^\top
\end{equation}

Because adding a row and adding a column are symmetrical, we describe our method in terms of adding a column to reduce redundant statements, and we use $\mathbf{B}$ to denote $\mathbf{B}_1$ and $\mathbf{B}_2$ above. 
The procedure for adding a row is the same, with all matrices are simply transposed.

In DAMF, without loss of generality, we consider the single node or edge update, that is, $\mathbf{B}$ (and $\mathbf{B}$, $\mathbf{C}$ in the Section \ref{sec:EdgeChange}) is $n$-by-$1$ matrix (or vector)

Due to the connectivity of graphs, adding nodes is often accompanied by adding edges. 
Let $\Delta m$ be the number of edges added to the graph. Since $\mathbf{B}$ is a part of the adjacency matrix, we give the following property without proof.

\begin{property}
$\mathbf{B}$ has at most $\Delta m$ non-zero elements.
\end{property}

For the matrix's expanded part $\mathbf{B}$, we use a space projection idea on the embedding space with a small number of modifications for node embeddings to fit the updated graph.

\begin{lemma}
\label{lemma:nzr}
For arbitrary matrices $\mathbf{B} \in \mathbb{R}^{n \times q}, \mathbf{C} \in \mathbb{R}^{q \times p}$, if $\mathbf{B}$ has $t$ non-zero rows, then $\mathbf{B}\mathbf{C}$ has at most $t$ non-zero rows.
\end{lemma}

\begin{theorem}[Space Projection for Node Change]
\label{theorem:node}
Assuming $\mathbf{B}$ is a matrix $\in \mathbb{R}^{n_1 \times 1}$ with at most $\Delta m$ non-zero elements. 
Let $\mathbf{X}_1 \in \mathbb{R}^{n_1 \times d}, \mathbf{Y}_1 \in \mathbb{R}^{n_2 \times d}$ be arbitrary network embedding  with 
\begin{equation}
\label{eq:svd1}
    \mathbf{X}_1^{}   \mathbf{Y}_1^\top = \mathbf{U}_1 \boldsymbol{\Sigma}_1 \mathbf{V}_1^\top = \widetilde{\mathbf{A}},
\end{equation}
where $\widetilde{\mathbf{A}} \in \mathbb{R}^{n_1 \times n_2}$, 

then there exists space projection matrices $\mathbf{F} \in \mathbb{R}^{d\times d}, \mathbf{G} \in \mathbb{R}^{d \times d}$, and embedding modification matrices $\Delta \mathbf{X}\in \mathbb{R}^{n_1 \times d}$, $\Delta \mathbf{Y} \in \mathbb{R}^{(n_2+1) \times d}$ with at most $\Delta m$ non-zero rows , such that
\begin{equation}
    \mathbf{X}_2 = \mathbf{X}_{1}   \mathbf{F} + \Delta \mathbf{X}, \quad
    \mathbf{Y}_2 = \mathbf{Y}_{1}   \mathbf{G} + \Delta \mathbf{Y},
\end{equation} where $\mathbf{X}_2 \in \mathbb{R}^{n_1 \times d}, \mathbf{Y}_2 \in \mathbb{R}^{(n_2+1) \times d}$ is a network embedding from a rank-$d$ t-SVD of $[ \widetilde{\mathbf{A}}, \mathbf{B} ]$, i.e.,
\begin{equation}
\label{eq:node_target}
    (\mathbf{U}_2, \boldsymbol{\Sigma}_2, \mathbf{V}_2) \gets \texttt{\textbf{t-SVD}}([ \widetilde{\mathbf{A}}, \mathbf{B} ], d)
\end{equation}
\begin{equation}
\label{eq:svd2}
     \mathbf{X}_2 = \mathbf{U}_2   \sqrt{\boldsymbol{\Sigma}_2}, \quad \mathbf{Y}_2 = \mathbf{V}_2   \sqrt{\boldsymbol{\Sigma}_2}
\end{equation}
\end{theorem}

\begin{proof}

Let $\mathbf{W}_\mathbf{B}=\mathbf{U}_1^\top \mathbf{B}$, and the normalized $(\mathbf{I}-\mathbf{U}_1\mathbf{U}_1^\top)\mathbf{B}$ be
\begin{equation}
\label{eq:qrnode}
R_\mathbf{B} = \Vert (\mathbf{I}-\mathbf{U}_1\mathbf{U}_1^\top) \mathbf{B} \Vert_2,\quad \mathbf{Q_B} = (\mathbf{I}-\mathbf{U}_1\mathbf{U}_1^\top) \mathbf{B} / R_\mathbf{B}
\end{equation} 
It can be proved that such a $\mathbf{Q}_\mathbf{B}$ is orthogonal to all column vectors of $\mathbf{U}_1$.

According to Zha-Simon's formula~\cite{zha1999updating}, we have
\begin{equation}
\label{eq:zhasimonnode1}
\begin{aligned}
\begin{bmatrix}
\mathbf{X}_1^{} \mathbf{Y}_1^\top & \mathbf{B}
\end{bmatrix}&=
\begin{bmatrix}
\mathbf{U}_1 \boldsymbol{\Sigma}_1 \mathbf{V}_1^\top & \mathbf{B}    
\end{bmatrix} \\ 
&=
\begin{bmatrix}
    \mathbf{U}_1 & \mathbf{Q_B}    
\end{bmatrix} 
\begin{bmatrix}
    \Sigma_1 & \mathbf{W_B} \\
      & R_\mathbf{B}
\end{bmatrix} 
\begin{bmatrix}
\mathbf{V}_1^\top & \\
& \mathbf{I}
\end{bmatrix} \\
& \approx ( 
\begin{bmatrix}
    \mathbf{U}_1 & \mathbf{Q_B}
\end{bmatrix}
  \mathbf{E} )   \boldsymbol{\Theta}   
(\begin{bmatrix}
    \mathbf{V}_1 & \\
    & \mathbf{I}
\end{bmatrix}   \mathbf{H} )^\top \\
& = \mathbf{U}_2  \boldsymbol{\Sigma}_2   \mathbf{V}_2^\top
\end{aligned}
\end{equation} 
with 
\begin{equation}
\label{eq:appd1}
\mathbf{U}_2=[\mathbf{U}_1\quad \mathbf{Q}_B]\mathbf{E}, \quad \Sigma_2=\boldsymbol{\Theta}, \quad \mathbf{V}_2=(\begin{bmatrix}
    \mathbf{V}_1 & \\
    & \mathbf{I}
\end{bmatrix}   \mathbf{H} )
\end{equation}
where the matrix product $\mathbf{E} \boldsymbol{\Theta} \mathbf{H}$ denotes a compact rank-$d$ t-SVD with
\begin{equation}
\label{eq:svd_small_node}
    \mathbf{E}, \boldsymbol{\Theta}, \mathbf{H} \gets \textbf{\texttt{t-SVD}}(\begin{bmatrix}
    \boldsymbol{\Sigma}_1 & \mathbf{W}_\mathbf{B}\\
      & R_\mathbf{B}
\end{bmatrix}, d)
\end{equation}

What is mentioned above is Zha-Simon's t-SVD update scheme.
The following will show how to obtain the space projection matrix and embedding modification vectors.

\emph{\textbf{The central idea of the proof is to express the update of 
$\ \mathbf{U}, \mathbf{V}$ in Eq.(\ref{eq:appd1}) as a space projection onto $\mathbf{U}, \mathbf{V}$ plus a matrix with at most $\Delta m$ non-zero rows.}} This is because in the update of $\mathbf{U}$ in Eq.(\ref{eq:appd1}), $\mathbf{Q}_\mathbf{B}$ can be written as 
\begin{equation}
\mathbf{Q}_\mathbf{B}=R_\mathbf{B}^{-1}(\mathbf{I}-\mathbf{U}_1\mathbf{U}_1^\top)\mathbf{B}
=R_\mathbf{B}^{-1}\mathbf{B}-\mathbf{U}_1(R_\mathbf{B}^{-1}\mathbf{U}_1^\top \mathbf{B})
\end{equation}
Notice that $R_\mathbf{B}$ is a scalar, so $R_\mathbf{B}^{-1} \mathbf{B}$ is a sparse matrix with at most $\Delta m$ non-zero rows. Moreover, $\mathbf{U}_1(R_\mathbf{B}^{-1} \mathbf{U}_1^\top \mathbf{B})$ is a space projection onto $\mathbf{U}_1$ (which can be further merged).

Specifically, we split the matrix multiplication in Eq.(\ref{eq:appd1}) for $U$ with
\begin{equation}
\begin{aligned}
\mathbf{U}_2 
&=[\mathbf{U}_1\quad \mathbf{Q}_B]\mathbf{E} \\
&=\mathbf{U}_1\mathbf{E}[:d] + \mathbf{Q}_\mathbf{B} \mathbf{E}[d:] \\
&=\mathbf{U}_1\mathbf{E}[:d] + (R_\mathbf{B}^{-1}\mathbf{B}-\mathbf{U}_1(R_\mathbf{B}^{-1}\mathbf{U}_1^\top \mathbf{B})\mathbf{E}[d:] \\
&=\mathbf{U}_1(\mathbf{E}[:d] - (R_\mathbf{B}^{-1}\mathbf{U}_1^\top \mathbf{B})\mathbf{E}[d:]) + R_\mathbf{B}^{-1} \mathbf{B} \mathbf{E}[d:]
\end{aligned}
\end{equation}
which is clearly a space projection on the columns of $\mathbf{U}_1$ plus a sparse matrix with at most $\Delta m$ non-zero rows. And, similarly, for the update on $\mathbf{V}$ in Eq.(\ref{eq:appd1}) we have
\begin{equation}
\mathbf{V}_2=(\begin{bmatrix}
    \mathbf{V}_1 & \\
    & \mathbf{I}
\end{bmatrix}   \mathbf{H} )
=\begin{bmatrix}
    \mathbf{V}_1\\
    \\
\end{bmatrix} \mathbf{H[:d]} + 
\begin{bmatrix}
    \\
    \mathbf{I}_{1\times 1}
\end{bmatrix} \mathbf{H[d:]}
\end{equation}

Since $\mathbf{X}_1=\mathbf{U}_1\sqrt{\mathbf{\Sigma}_1}$, $\mathbf{Y}_1=\mathbf{V}_1 \sqrt {\mathbf{\Sigma}_1}$ and $\mathbf{X}_2=\mathbf{U}_2\sqrt{\mathbf{\Sigma}_2}$, $\mathbf{Y}_2=\mathbf{V}_2 \sqrt {\mathbf{\Sigma}_2}$, we have

\begin{equation}
\begin{aligned}
\label{eq:fgxynode}
    \mathbf{X}_2 =& \mathbf{X}_1   \boldsymbol{\Sigma}_1^{-1/2}   (\mathbf{E}[:d]-R_\mathbf{B}^{-1} \mathbf{W}_\mathbf{B} \mathbf{E}[d:])  \boldsymbol{\Sigma}_2^{1/2} \\ 
    & + R_{\mathbf{B}}^{-1} \mathbf{B} \mathbf{E}[d:]   \boldsymbol{\Sigma}_2^{1/2}  \\
    \mathbf{Y}_2 =& \mathbf{Y}_1   \boldsymbol{\Sigma}_1^{-1/2}   \mathbf{H}[:d]  \boldsymbol{\Sigma}_2^{1/2}  \\
    & + \begin{bmatrix}
    \\
    \mathbf{I}_{1 \times 1} \\
    \end{bmatrix}   \mathbf{H}[d:]   \boldsymbol{\Sigma}_2^{1/2} 
\end{aligned}
\end{equation} 
and we take
\begin{equation}
    \label{eq:f_node}
    \mathbf{F} = \boldsymbol{\Sigma}_1^{-1/2}   (\mathbf{E}[:d]-R_\mathbf{B}^{-1} \mathbf{W}_\mathbf{B}\mathbf{E}[d:])  \boldsymbol{\Sigma}^{1/2}_2
\end{equation}
\begin{equation}
    \label{eq:g_node}
    \mathbf{G} = \boldsymbol{\Sigma}_1^{-1/2}   \mathbf{H}[:d]   \boldsymbol{\Sigma}^{1/2}_2
\end{equation}
\begin{equation}
    \label{eq:x_node}
    \Delta \mathbf{X} = R_{\mathbf{B}}^{-1} \mathbf{B} \mathbf{E}[d:]   \boldsymbol{\Sigma}_2^{1/2}
\end{equation}
\begin{equation}
    \label{eq:y_node}
    \Delta \mathbf{Y} = \begin{bmatrix}
         \\
        \mathbf{I}_{1 \times 1} \\
    \end{bmatrix}   \mathbf{H}[d:]   \boldsymbol{\Sigma}^{1/2}_2
\end{equation} 

Since $\mathbf{B}$ has at most $\Delta m$ non-zero rows (as stated in Lemma \ref{lemma:nzr}), $\Delta \mathbf{X}$ has at most $\Delta m$ non-zero rows, and it is clear that $\Delta \mathbf{Y}$ has only one non-zero row.
\end{proof}
The above proof provides the space projection matrices $\mathbf{F}, \mathbf{G}$ to use in the embedding space, and shows that $\Delta \mathbf{X}$ and $\Delta \mathbf{Y}$ require at most $\Delta m$ nodes' embedding to be modified additionally in a node change process.

%% file: Methodology/3EdgeChange.tex
\subsection{Edge Change}
\label{sec:EdgeChange}
Considering adding a directed edge $(u,v)$ to the graph, the change in the adjacency matrix can be expressed as a low-rank update by
\begin{equation}
\mathbf{A}_t \gets \mathbf{A}_{t-1} + \Delta \mathbf{A}_t, \quad 
\Delta \mathbf{A}_{t-1} = \mathbf{B}\mathbf{C}^\top
\end{equation}
 with 
\begin{equation}
\label{eq:edge_B}
 \mathbf{B}=e_u,\quad \mathbf{C}=e_v
\end{equation} where $e_i$ denotes the standard basis (i.e., a vector whose components are all zero, except the $i$-th element is 1).

\begin{theorem}[Space Projection for Edge Change]
\label{theorem:edge}
Assuming $\mathbf{B}$ and $\mathbf{C}$ are matrices $\in \mathbb{R}^{n \times 1}$ with at most $\Delta m$ non-zero elements. 
Let $\mathbf{X}_1 \in \mathbb{R}^{n \times d}, \mathbf{Y}_1 \in \mathbb{R}^{n \times d}$ be arbitrary network embedding with 
\begin{equation}
    \mathbf{X}_1^{}   \mathbf{Y}_1^\top = \mathbf{U}_1 \boldsymbol{\Sigma}_1 \mathbf{V}_1^\top = \widetilde{\mathbf{A}},
\end{equation}
where $\widetilde{\mathbf{A}} \in \mathbb{R}^{n \times n}$, 

then there exists space projection matrices $\mathbf{F} \in \mathbb{R}^{d\times d}, \mathbf{G} \in \mathbb{R}^{d \times d}$, and embedding modification vectors $\Delta \mathbf{X}\in \mathbb{R}^{n \times d}$ and $\Delta \mathbf{Y} \in \mathbb{R}^{n \times d}$ with at most $\Delta m$ non-zero rows , such that
\begin{equation}
    \mathbf{X}_2 = \mathbf{X}_{1}   \mathbf{F} + \Delta \mathbf{X}, \quad
    \mathbf{Y}_2 = \mathbf{Y}_{1}   \mathbf{G} + \Delta \mathbf{Y},
\end{equation} where $\mathbf{X}_2 \in \mathbb{R}^{n \times d}, \mathbf{Y}_2 \in \mathbb{R}^{n \times d}$ is a network embedding from a rank-$d$ t-SVD of $\widetilde{\mathbf{A}} + \mathbf{B} \mathbf{C} ^ \top $, i.e.,
\begin{equation}
\label{eq:node_target1}
    (\mathbf{U}_2, \boldsymbol{\Sigma}_2, \mathbf{V}_2) \gets \texttt{\textbf{t-SVD}}(\widetilde{\mathbf{A}} + \mathbf{B} \mathbf{C}^\top, d)
\end{equation}
\begin{equation}
\label{eq:svd_edge_2}
     \mathbf{X}_2 = \mathbf{U}_2   \sqrt{\boldsymbol{\Sigma}_2}, \quad \mathbf{Y}_2 = \mathbf{V}_2   \sqrt{\boldsymbol{\Sigma}_2}
\end{equation}
\end{theorem}
\begin{proof}

Similar to the proof of Theorem \ref{theorem:node}, firstly, let $\mathbf{W}_\mathbf{B}=\mathbf{U}_1^\top B$ and $\mathbf{W}_\mathbf{C}=\mathbf{V}_1^\top \mathbf{C}$. 
And the normalized $(\mathbf{I}-\mathbf{U}_1\mathbf{U}_1^\top) \mathbf{B}$, $(\mathbf{I}-\mathbf{V}_1\mathbf{V}_1^\top) \mathbf{C}$ can be calculated by
\begin{equation}
\label{eq:qredge1}
R_\mathbf{B} = \Vert (\mathbf{I}-\mathbf{U}_1\mathbf{U}_1^\top) \mathbf{B} \Vert_2,
\quad
\mathbf{Q_B} = (\mathbf{I}-\mathbf{U}_1\mathbf{U}_1^\top) \mathbf{B} / R_\mathbf{B}
\end{equation}
\begin{equation}
\label{eq:qredge2}
R_\mathbf{C} = \Vert (\mathbf{I}-\mathbf{V}_1\mathbf{V}_1^\top) \mathbf{C} \Vert_2,
\quad
\mathbf{Q_C} = (\mathbf{I}-\mathbf{V}_1\mathbf{V}_1^\top) \mathbf{C} / R_\mathbf{C}
\end{equation}

Then let 
\begin{equation}
\label{eq:svd_small_edge}
    \mathbf{E},\mathbf{\Sigma}_2, \mathbf{H} \gets \texttt{\textbf{t-SVD}}(
    \begin{bmatrix}
        \Sigma_1 & \mathbf{0} \\
        \mathbf{0} & \mathbf{0}
    \end{bmatrix}
    +
    \begin{bmatrix}
        \mathbf{W}_\mathbf{B} \\
        R_\mathbf{B}
    \end{bmatrix}
    \begin{bmatrix}
        \mathbf{W}_\mathbf{C} \\
        R_\mathbf{C}
    \end{bmatrix}^\top, d)
\end{equation}
We can get the space projection matrices $\mathbf{F}, \mathbf{G}$ and embedding modification vectors $\Delta \mathbf{X}$, $\Delta \mathbf{Y}$ by
\begin{equation}
    \label{eq:f_edge}
    \mathbf{F} = \boldsymbol{\Sigma}_1^{-1/2}   (\mathbf{E}[:d]- R_\mathbf{B}^{-1} \mathbf{W}_\mathbf{B} \mathbf{E}[d:])  \boldsymbol{\Sigma}^{1/2}_2
\end{equation}
\begin{equation}
    \label{eq:g_edge}
    \mathbf{G} = \boldsymbol{\Sigma}_1^{-1/2}   (\mathbf{H}[:d] - R_\mathbf{C}^{-1} \mathbf{W}_\mathbf{C} \mathbf{H}[d:])   \boldsymbol{\Sigma}^{1/2}_2
\end{equation}
\begin{equation}
\label{eq:x_edge}
    \Delta \mathbf{X} = R_{\mathbf{B}}^{-1} \mathbf{B} \mathbf{E}[d:] \boldsymbol{\Sigma}_2^{1/2}
\end{equation}
\begin{equation}
\label{eq:y_edge}
    \Delta \mathbf{Y} = R_{\mathbf{C}}^{-1} \mathbf{C} \mathbf{H}[d:] \boldsymbol{\Sigma}_2^{1/2}
\end{equation} 

Since $\mathbf{B}, \mathbf{C}$ have at most $\Delta m$ non-zero rows, by Lemma \ref{lemma:nzr}, $\Delta \mathbf{X}, \Delta \mathbf{Y}$ have at most $\Delta m$ non-zero rows.
\end{proof}
The above proof provides the space projection matrices $\mathbf{F}, \mathbf{G}$ to use in the embedding space, and shows that $\Delta \mathbf{X}$ and $\Delta \mathbf{Y}$ require at most $\Delta m$ nodes' embedding to be modified additionally in an edge change process.

%% file: Methodology/4PageRank.tex
\subsection{Dynamic Embedding Enhancement via PPR}
\label{sec:pagerank}
In order to capture higher-order neighbors's information, an enhancement is applied to the dynamic network embedding.
Specifically, we apply a dynamic \emph{Personalized PageRank} (PPR)\cite{page1999pagerank} to the updated context embedding $\mathbf{X}$ to get the enhanced context embedding $\mathbf{Z}$. 

\textbf{PPR Enhancement in Static Network Embedding.} To better capture higher-order neighborhood information, a mainstream approach on static network embeddings is to use \emph{Personalized PageRank} (PPR)~\cite{page1999pagerank} to enhance the network embeddings (e.g., APP, Lemane, STRAP, VERSE, and NRP). 

Specifically, for a graph with adjacency matrix $\mathbf{A}$ and graph signal $\mathbf{X}\in \mathbb{R}^{n\times d}$, the principles of PPR can be formulated as
\begin{equation}
\label{eq:ppr}
    \texttt{PPR}(\mathbf{X})=\sum_{i=0}^\infty \alpha(1-\alpha)^i (\mathbf{D}^{-1} \mathbf{A})^i \mathbf{X}
\end{equation}
where $\mathbf{D} \in \mathbb{R}^{n\times n}$ is the diagonal out-degree matrix and $\alpha$ is the damping factor for PageRank. 

Since solving the PPR is an infinite process, the truncated version of PPR is commonly used in practice. Moreover, to better control errors and balance efficiency, existing PageRank-based static network embedding methods (e.g., all methods mentioned above) typically introduce an error tolerance $\epsilon$ and require PPR to converge to that error.

\textbf{Dynamic Embedding Enhancement via Dynamic PPR.}
Unfortunately, despite the success of PageRank enhancements in static network embeddings, their methods cannot be directly migrated to dynamic scenarios. In this work, following \cite{bojchevski2020scaling, wang2021approximate, zheng2022instant, zhang2016approximate}, we approximate the PPR based on graph propagation under a given error limit $\epsilon$. 

Specifically, we use an InstantGNN~\cite{zheng2022instant} to enhance the embedding obtained in Section \ref{sec:NodeChange} and Section \ref{sec:EdgeChange}. InstantGNN is an efficient dynamic PPR method suitable for graph structure changes and node signals (attributes).
In Dynamic Embedding Enhancement, we use the embedding obtained from dynamic truncated singular value decomposition as the graph signal input to InstantGNN. Moreover, as the network changes, according to Theorem \ref{theorem:node} and Theorem \ref{theorem:edge}, at most $\Delta m$ nodes' signals need to be changed. The 
InstantGNN updates the result, which converges to an error tolerance $\epsilon$ according to the graph structure and node signal change to achieve the Dynamic Embedding Enhancement.

%% file: Methodology/5ProposedMethod.tex
\begin{algorithm}[tp]
  \LinesNumbered 
  \SetKwFunction{UpdateEmbeddingN}{UpdateEmbeddingN}
  \SetKwFunction{UpdateEmbeddingE}{UpdateEmbeddingE}
  \SetKwProg{myproc}{Procedure}{}{}
  
  \myproc{\UpdateEmbeddingN{$\mathbf{X}, \mathbf{Y}, \mathbf{P_X}, \mathbf{P_Y}, \mathbf{\Sigma}_1, \mathbf{B}, d$}}{
    \tcc{Step 1: Orthogonalization}
    $\mathbf{W}_\mathbf{B} \gets (\mathbf{B}^\top \mathbf{X} \mathbf{P_X} \mathbf{\Sigma}_1^{1/2})^\top $ \;
    $R_\mathbf{B} \gets \sqrt{ \Vert \mathbf{B} \Vert_2^2 - \Vert \mathbf{W}_\mathbf{B} \Vert_2^2 }$\;
    \tcc{Step 2: Rediagonalization}
    $\mathbf{M} \gets \begin{bmatrix}
    \Sigma_1 & \mathbf{W}_\mathbf{B} \\
      & R_\mathbf{B} \end{bmatrix}$
      \tcp*{$\mathbf{M} \in \mathbb{R}^{d+1, d+1}$}
    $\mathbf{E}, \Sigma_2, \mathbf{H} \gets $ \texttt{t-SVD}($\mathbf{M}, d$) \;
    \tcc{Step 3: Space Projection}
    $\mathbf{F} \gets \boldsymbol{\Sigma}_1^{-1/2}  (\mathbf{E}[:d]-R_\mathbf{B}^{-1} \mathbf{W}_\mathbf{B}   \mathbf{E}[d:]) \boldsymbol{\Sigma}^{1/2}_2$ \;
    $\mathbf{G} \gets \boldsymbol{\Sigma}_1^{-1/2}    \mathbf{H}[:d]    \boldsymbol{\Sigma}^{1/2}_2 $ \;
    $\Delta \mathbf{X} = R_\mathbf{B}^{-1} \mathbf{B}   \mathbf{E}[d:]   \boldsymbol{\Sigma}_2^{1/2}$ \;
    $\Delta \mathbf{Y} =   \begin{bmatrix}
        \\
        \mathbf{I}_{1 \times 1} \\
    \end{bmatrix}  \mathbf{H}[d:]    \boldsymbol{\Sigma}^{1/2}_2$ \;
    
   \KwRet{$\mathbf{F}, \mathbf{G}, \Delta \mathbf{ X}, \Delta \mathbf{Y}, \Sigma_2$ }\;}
  \SetKwProg{myproc}{Procedure}{}{}
  \myproc{\UpdateEmbeddingE{$\mathbf{X}, \mathbf{Y}, \mathbf{P_X}, \mathbf{P_Y}, \mathbf{\Sigma}_1, \mathbf{B}, \mathbf{C}, d$}}{
    \tcc{Step 1: Orthogonalization}
    $\mathbf{W}_\mathbf{B} \gets (\mathbf{B}^\top \mathbf{X} \mathbf{P_X} \mathbf{\Sigma}_1^{1/2})^\top $ \;
    $\mathbf{W}_\mathbf{C} \gets (\mathbf{C}^\top \mathbf{Y} \mathbf{P_Y} \mathbf{\Sigma}_1^{1/2})^\top $ \;
    $R_\mathbf{B} \gets \sqrt{ \Vert \mathbf{B} \Vert_2^2 - \Vert \mathbf{W}_\mathbf{B} \Vert_2^2 }$\;
    $R_\mathbf{C} \gets \sqrt{ \Vert \mathbf{C} \Vert_2^2 - \Vert \mathbf{W}_\mathbf{C} \Vert_2^2 }$\;

    \tcc{Step 2: Rediagonalization}
    $\mathbf{M} \gets\begin{bmatrix}
            \Sigma_1 & \mathbf{0} \\
            \mathbf{0} & \mathbf{0}
        \end{bmatrix}
        +
        \begin{bmatrix}
            \mathbf{W}_\mathbf{B} \\
            \mathbf{R}_\mathbf{B}
        \end{bmatrix}
        \begin{bmatrix}
            \mathbf{W}_\mathbf{C} \\
            \mathbf{R}_\mathbf{C}
        \end{bmatrix}^\top  $ \;
    $\mathbf{E}, \Sigma_2, \mathbf{H} \gets $ \texttt{t-SVD}($\mathbf{M}, d$) \;
    
    \tcc{Step 3: Space Projection}
    $\mathbf{F} \gets \boldsymbol{\Sigma}_1^{-1/2}   (\mathbf{E}[:d]-R_\mathbf{B}^{-1}\mathbf{W}_\mathbf{B}   \mathbf{E}[d:])  \boldsymbol{\Sigma}^{\frac{1}{2}}_2$ \; 
    $\mathbf{G} \gets \boldsymbol{\Sigma}_1^{-1/2}   (\mathbf{H}[:d]-R_\mathbf{C}^{-1}\mathbf{W}_\mathbf{C}   \mathbf{H}[d:])   \boldsymbol{\Sigma}^{1/2}_2$ \;
    $\Delta \mathbf{X} \gets R_\mathbf{B}^{-1} \mathbf{B}   \mathbf{E}[d:]   \boldsymbol{\Sigma}_2^{1/2}$ \;
    $\Delta \mathbf{Y} \gets R_\mathbf{C}^{-1} \mathbf{C}   \mathbf{H}[d:]   \boldsymbol{\Sigma}_2^{1/2}$ \;
    \KwRet{$\mathbf{F}, \mathbf{G}, \Delta \mathbf{X}, \Delta \mathbf{Y}, \Sigma_2$ }\;}
  \caption{Update embedding via space projection}
  \label{algo:sp}

\end{algorithm}

\subsection{the Proposed DAMF Algorithm}
\label{sec:damf}
\begin{algorithm}[t] 
    \caption{DAMF}
    \label{algo:damf}
    \LinesNumbered 
    \SetKwFunction{UpdateEmbeddingN}{UpdateEmebddingN}
    \SetKwFunction{UpdateEmbeddingM}{UpdateEmbeddingM}
    \SetKwFunction{DynamicEnhancement}{DynamicEnhancement}
    \KwIn{Network embedding $\mathbf{X}_{b}, \mathbf{Y}_{b}$ in the base space;
    enhanced embedding $\mathbf{Z}_{b}$; space projection matrix $\mathbf{P}_{\mathbf{X}}, \mathbf{P}_{\mathbf{Y}}$; singular values $\boldsymbol{{\Sigma}}$; the Changes occurring in the network $Event_t$ residual vector $\mathbf{r}$; PageRank damping factor $\alpha$; error tolerance $\epsilon$; Graph $\mathcal{G}$; }
    \KwOut{Updated $\mathbf{X}_{b}, \mathbf{Y}_{b}, \mathbf{Z}_{b}, \mathbf{P}_{\mathbf{X}},  \mathbf{P}_{\mathbf{Y}}, \boldsymbol{{\Sigma}}$}

    \eIf{$Event_t$ is node change}{
        Convert $Event_t$ to $\mathbf{B}_1, \mathbf{B}_2$ by Eq. (\ref{eq:node_B}) \;
        $\mathbf{F}, \mathbf{G}, \Delta \mathbf{X}, \Delta \mathbf{Y} , \Sigma \gets$ \UpdateEmbeddingN{$\mathbf{X}, \mathbf{Y}, \mathbf{\Sigma}, \mathbf{B}_1$} \;
        Update $\mathbf{X}_{b}, \mathbf{Y}_{b}, \mathbf{P}_{\mathbf{X}}, \mathbf{P}_{\mathbf{Y}}$ by Eq. (\ref{eq:update}) \;
        $\mathbf{F}, \mathbf{G}, \Delta \mathbf{X}, \Delta \mathbf{Y} , \Sigma \gets$ \UpdateEmbeddingN{$\mathbf{Y}, \mathbf{X}, \mathbf{\Sigma}, \mathbf{B}_2$}  \;
        Update $\mathbf{X}_{b}, \mathbf{Y}_{b}, \mathbf{P}_{\mathbf{X}}, \mathbf{P}_{\mathbf{Y}}$ by Eq. (\ref{eq:update}) \;
    }{
        Convert $Event_t$ to $\mathbf{B}, \mathbf{C}$ by Eq. (\ref{eq:edge_B}) \;
        $\mathbf{F}, \mathbf{G}, \Delta \mathbf{X}, \Delta \mathbf{Y} , \Sigma \gets$ \UpdateEmbeddingM{$\mathbf{X}, \mathbf{Y}, \mathbf{\Sigma}, \mathbf{B}, \mathbf{C}$}\;
        Update $\mathbf{X}_{b}, \mathbf{Y}_{b}, \mathbf{P}_{\mathbf{X}}, \mathbf{P}_{\mathbf{Y}}$ by Eq. (\ref{eq:update}) \;
    }
    $\mathbf{Z}_{b}, \mathbf{r} \gets $ \DynamicEnhancement{$\mathbf{Z}_{b}, \mathbf{r}, Event_t, \mathcal{G}, \Delta \mathbf{X}, \alpha, \epsilon$)}\;
    \KwRet{$\mathbf{X}_{b}, \mathbf{Y}_{b}, \mathbf{Z}_{b}, \mathbf{P}_{\mathbf{X}},  \mathbf{P}_{\mathbf{Y}}, \boldsymbol{{\Sigma}}$}
\end{algorithm}

In this section, we propose the Dynamic Adjacency Matrix Factorization (DAMF) algorithm, consisting of four steps: Orthogonalization, Rediagonalization, Space Projection, and Dynamic Embedding Enhancement.
Algorithm \ref{algo:sp} and Algorithm \ref{algo:damf} is the pseudo-code for the DAMF algorithm.

\textbf{Step 1: Orthogonalization.}
According to the equations $\mathbf{X} = \mathbf{U} \sqrt{\Sigma}$ and $\mathbf{X}_t=\mathbf{X}_b\mathbf{P_X}$, we know that $\mathbf{U}=\mathbf{X}_b \mathbf{P_X} \Sigma^{-1/2}$. 
Then, with $\mathbf{W_B} = \mathbf{U}^\top\mathbf{B}$ (and $\mathbf{W_C} = \mathbf{V}^\top\mathbf{C}$ for edge change), we can get that $R_B=\sqrt{ \Vert \mathbf{B}\Vert_2^2 - \Vert \mathbf{W}_\mathbf{B} \Vert_2^2 }$, (and $R_C=\sqrt{ \Vert \mathbf{C}\Vert_2^2 - \Vert \mathbf{W}_\mathbf{C} \Vert_2^2 }$ for edge change) by Eq. (\ref{eq:qrnode}) and Eq. (\ref{eq:qredge2}) and Proposition \ref{proposition:fast}.

\begin{proposition}
    \label{proposition:fast}
    Let $\mathbf{U}\in \mathbb{R}^{n\times d}$ be an arbitrary orthonormal matrix with $\mathbf{U}^\top\mathbf{U}=\mathbf{I}$. and $\vec{x} \in \mathbb{R}^{n}$ be an arbitrary vector, then 
    $$\Vert (\mathbf{I}-\mathbf{U} \mathbf{U}^\top)\vec{x} \Vert_2 = \sqrt{ \Vert \vec{x}\Vert_2^2 - \Vert \mathbf{U}^\top \vec{x}\Vert_2^2 }$$
\end{proposition}

\textbf{Step 2: Rediagonalization.}
In this step, follow Eq. (\ref{eq:svd_small_node}) for node change, and Eq. (\ref{eq:svd_small_edge}) for edge change to get the t-SVD.

\textbf{Step 3: Space Rotation.}
For node change, get $\mathbf{F}, \mathbf{G}, \Delta \mathbf{X}, \Delta \mathbf{Y}$ by Eq. (\ref{eq:f_node}), Eq. (\ref{eq:g_node}), Eq. (\ref{eq:x_node}), Eq. (\ref{eq:y_node}), respectively. And for edge change, get $\mathbf{F}, \mathbf{G}, \Delta \mathbf{X}, \Delta \mathbf{Y}$ by Eq. (\ref{eq:f_edge}), Eq. (\ref{eq:g_edge}), Eq. (\ref{eq:x_edge}), Eq. (\ref{eq:y_edge}), respectively. Then update the network embedding with space projection matrix $\mathbf{F}, \mathbf{G}, \Delta \mathbf{X}, \Delta \mathbf{Y}$ by Eq. (\ref{eq:update}).

\textbf{Step 4: Dynamic Embedding Enhancement.} 
In this step, the PPR-enhanced embedding is updated by InstantGNN~\cite{zheng2022instant}.
Specifically, InstantGNN can estimate the PPR values for dynamic graph structures and dynamic node attributes.
Here, we treat the changes in the graph as the dynamic graph structure, and the updates of few nodes' embedding as the dynamic node attributes.
InstantGNN will update the embedding so that the results converge to the given error tolerance.

\textbf{Initialization of DAMF.}
Initially, $\mathbf{X}_b, \mathbf{Y}_b$ is set by a t-SVD of the adjacency matrix $\mathbf{A}$ of the initial graph, and $\mathbf{P_X}, \mathbf{P_y}$ is set to the identity matrix. 
Then, use the basic propagation of PPR in InstantGNN~\cite{zheng2022instant} to enhance the initial network embedding.

%% file: Methodology/6ComplexityAnalysis.tex
\subsection{Complexity Analysis}
\label{sec:complexity}

In this section, we analyze the DAMF algorithm in terms of time and space complexity.
Note that the time complexity we analysis here is only for a single graph change.

DAMF is very efficient due to the fact that the large-scale matrix multiplications involved in DAMF are mainly of two kinds showed in Figure \ref{fig:matrixmultiply}. And Lemma \ref{lemma:complexity_A} (Figure \ref{fig:matrixmultiply}(a)) and Lemma \ref{lemma:complexity_B} ((Figure \ref{fig:matrixmultiply}(b))) demonstrate how to use sparsity of these two special matrix multiplications and how they can be efficiently computed.
The proofs of these two lemmas are in the Appendix \ref{appendix:complexity_A} and Appendix \ref{appendix:complexity_B}.

\begin{figure}[htbp]
	\centering
        \setlength{\abovecaptionskip}{0.cm}
        \setlength{\belowcaptionskip}{-0.cm}
	\subfigure[\label{fig:a}]{
		\includegraphics[scale=2]{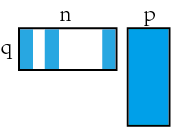}}
        \quad \quad
	\subfigure[\label{fig:c}]{
		\includegraphics[scale=2]{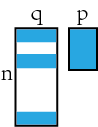}}
	\caption{Two special matrix multiplications that can be efficiently computed as proved by Lemma 2 and Lemma 3 (the white color indicates zero elements in matrices)}
        \label{fig:matrixmultiply}
\end{figure}

\begin{lemma}
    \label{lemma:complexity_A}
    Let $\mathbf{A} \in \mathbb{R}^{n\times p}, \mathbf{B}\in \mathbb{R}^{n\times q}$ be arbitrary matrices with $\mathbf{B}$ has $t$ non-zero rows, the time complexity to calculate $\mathbf{B}^\top \mathbf{A}$ is $O(tpq)$.
\end{lemma}
\begin{lemma}
    \label{lemma:complexity_B}
    Let $\mathbf{B}\in \mathbb{R}^{n\times q}, \mathbf{C} \in \mathbb{R}^{q \times p}$ be arbitrary matrices with $\mathbf{B}$ has $t$ non-zero rows, the time complexity to calculate $\mathbf{BC}$ is $O(tpq)$. 
\end{lemma}

\textbf{Time Complexity of Orthogonalization (Step 1). }
According to Lemma \ref{lemma:complexity_A}, the time complexity of calculating $\mathbf{B}^\top \mathbf{X}$ (or $\mathbf{C}^\top \mathbf{Y}$) is $O((\Delta m) d^2)$ since $\mathbf{B}$ (and $\mathbf{C}$) has at most $\Delta m$ non-zero rows.
Next, since $\mathbf{P_X}$ and $\Sigma_1^{1/2}$ are $d$-by-$d$ matrices, the complexity of calculating $\mathbf{W_B}$ (or $\mathbf{W_C}$, Line 2, 12, 13 in Algorithm \ref{algo:sp}) is $O((\Delta m)d^2)$. 
At last, the time complexity of calculating $R_\mathbf{B}$(or $R_\mathbf{C}$, Line 3, 14, 15 in Algorithm \ref{algo:sp}) is $O(\Delta m + d)$. 
Thus, the overall time complexity of the Orthogonalization step is $O((\Delta m)d^2)$.

\textbf{Time Complexity of Rediagonalization (Step 2). }
The time complexity of the Rediagonalization step is $O(d^3)$ by directly applying a t-SVD on a $(d+1)$-by-$(d+1)$ matrix (Line 4, 5, 16, 17 in Algorithm \ref{algo:sp}).

\textbf{Time Complexity of Space Rotation (Step 3). }
The time complexity of calculating $\mathbf{F}$ and $\mathbf{G}$ (Line 6, 7, 18, 19 in Algorithm \ref{algo:sp}) is $O(d^3)$.
According to Lemma \ref{lemma:complexity_B}, the time complexity of calculating $\Delta \mathbf{X}$ and $\Delta \mathbf{Y}$ (Line 8, 9, 20, 21 in Algorithm \ref{algo:sp}) is $O((\Delta m) d^2)$.

\begin{figure*}[t]
\setlength{\abovecaptionskip}{0.cm}
\centering 
\includegraphics[height=4.2cm,width=18cm]{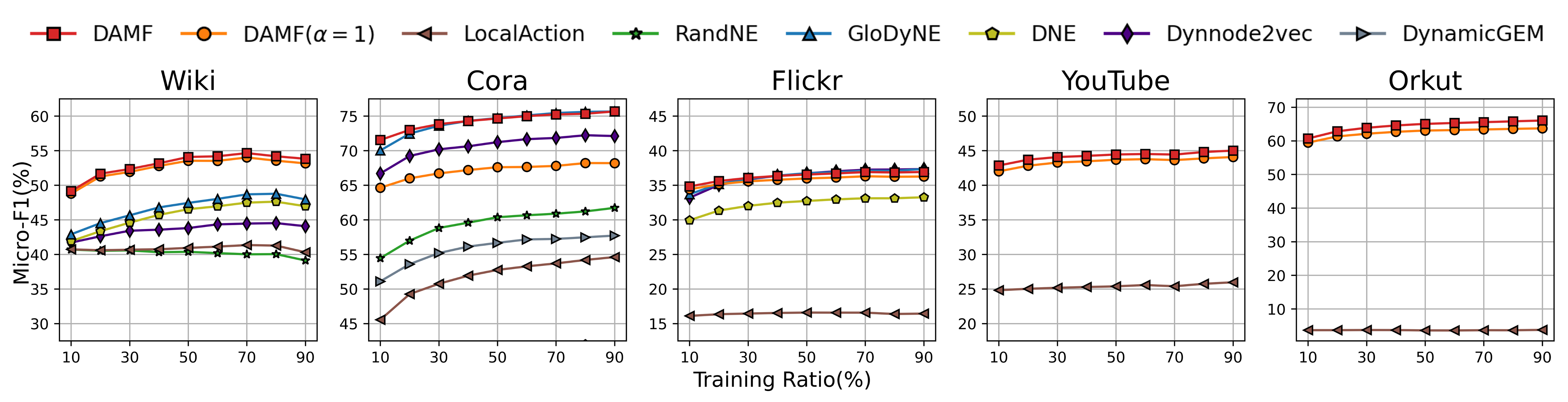}
\caption{Node classification's predictive performance w.r.t. the ratio of training data}
\label{fig:nc}
\vspace{-1em}
\end{figure*}

Moreover, the time complexity of space projection (Line 4, 6, 10 in Algorithm \ref{algo:damf}) is $O((\Delta m)d^2+d^3)$ by the following steps: (1) computing the update of $\mathbf{P_X}$ by a $d$-by-$d$ matrix multiplication; (2) computing the inverse of $\mathbf{P_X}$; (3) computing $\mathbf{BP}^{-1}_\mathbf{X}$. The time complexity of both the matrix multiplication and the inverse is $O(d^3)$, while the time complexity of $\mathbf{BP}^{-1}_\mathbf{X}$ is $O((\Delta m)d^2)$ according to Lemma \ref{lemma:complexity_B}.

\textbf{Time Complexity of Dynamic Embedding Enhancement (Step 4).}
Let $c=\max_{u\in \mathcal{V}}\Vert deg(u) \mathbf{X}(u) \Vert_\infty$ where $\mathbf{X}[u]$ is the node $u$'s context embedding. Since the number of edges changed by the graph structure is $\Delta m$, and at most $\Delta m$ nodes' context embedding that need to be updated in the previous step, according to Theorem 3 and Theorem 4 in \cite{zheng2022instant}, the time complexity of this step is $O(\frac{(\Delta m)cd}{\alpha^2 \epsilon})$.


Overall, if we consider embedding dimension $d$, Personalized PageRank damping factor $\alpha$, error tolerance $\epsilon$ and $c$ as constant, the purposed DAMF algorithm achieves a dynamic network embedding for a single graph change with time complexity $O(\Delta m)$.

\textbf{Time Complexity of Query Node's Embedding}
The node embedding query at any moment can be obtained by multiplying the embedding $\mathbf{X}_b, \mathbf{Y}_b$ of that node in the initial space by the space projection matrix $\mathbf{P}_\mathbf{X}, \mathbf{P}_\mathbf{y}$ with time complexity of $O(d^2)$, while the direct query in the static embedding method has a time complexity of $O(d)$. When we consider dimensionality a constant, the query complexity is $O(1)$.

\begin{table}[t]
\caption{Time Complexity of DAMF}
\begin{tabular}{c|c|c}
\toprule
                                 & \textbf{Name}     & \textbf{Time Complexity}                                     \\ \hline
\textbf{ Step 1 }                  & Orthogonalization  & $O((\Delta m) d^2)$                                          \\ \hline
\textbf{ Step 2 }                  & Rediagonalization & $O(d^3)$                                                     \\ \hline
\textbf{ Step 3 }                  & Space Rotation    & $O((\Delta m) d^2 + d^3)$                                    \\ \hline
\multirow{2}{*}{\textbf{ Step 4 }} & Dynamic Embedding & \multirow{2}{*}{$O(\frac{(\Delta m)cd}{\alpha^2 \epsilon})$} \\
                                 & Enhancement       &                                                              \\ \hline
\end{tabular}
\end{table}

\textbf{Space Complexity.} 
The largest matrices in the DAMF method is an $n$-by-$d$ matrix, so the space complexity is $O(nd)$.
Due to the additional need to store the structure of the graph, the total space complexity is $O(nd + m)$.

%% file: 05Experiment.tex
\begin{table}[h]
\caption{Statistics of Datasets}
\label{tab:dataset}
\begin{tabular}{cc|c|c|c}
\toprule
\multicolumn{2}{c|}{\textbf{Dataset}}                 & $\mathbf{|\mathcal{V}|}$ & $\mathbf{|\mathcal{E}|}$ & \textbf{\#labels} \\ \hline
\multicolumn{1}{c|}{\multirow{3}{*}{small}} & Wiki    & 4,777                     & 184,812                  & 40                \\ \cline{2-5} 
\multicolumn{1}{c|}{}                       & Cora    & 12,022                   & 45,421                   & 10                \\ \cline{2-5} 
\multicolumn{1}{c|}{}                       & Flickr  & 80,513                   & 11,799,764               & 195               \\ \hline
\multicolumn{1}{c|}{\multirow{2}{*}{large}} & YouTube & 1,138,499                & 2,990,443                & 47                \\ \cline{2-5} 
\multicolumn{1}{c|}{}                       & Orkut   & 3,072,441                & 117,185,083              & 100               \\ \hline
\multicolumn{1}{c|}{massive}                & Twitter & 41,652,230               & \textbf{1,468,365,182}   & -                 \\ \bottomrule
\end{tabular}
\end{table}
\section{Experiment}


In this section, we experimentally compare DAMF to six existing methods on six datasets of varying sizes for three popular analytic tasks: node classification, link prediction, and graph reconstruction. Section \ref{sec:settings} and Section \ref{sec:tasks} introduce the experimental settings and tasks.
In Section \ref{sec:smallg}, Section \ref{sec:largeg}, and Section \ref{sec:massive}, we present our experimental results on small, large, and massive graphs respectively. Finally, we compare the runtime of DAMF with baselines in Section \ref{sec:runtime} and analyze the results of the ablation study on dynamic embedding enhancement in Section \ref{sec:ablation}.

\subsection{Experimental Settings}

\label{sec:settings}
\textbf{Baseline}.
We evaluate DAMF against six existing methods, including LocalAction~\cite{liu2019real}, GloDyNE~\cite{hou2020glodyne}, Dynnode2vec~\cite{mahdavi2018dynnode2vec}, RandNE~\cite{zhang2018billion}, DynGEM~\cite{goyal2018dyngem} and DNE~\cite{du2018dynamic}.

\noindent\textbf{Datasets}. 
We conduct experiments on 6 publicly available graph datasets and divide the datasets into small, large, and massive scales.
As shown in Table \ref{tab:dataset}, small ones include \emph{Wiki}~\cite{wiki}, \emph{Cora}~\cite{hou2020glodyne} and \emph{Flickr}~\cite{flickr}, large ones include \emph{YouTube}~\cite{youtube} and \emph{Orkut}~\cite{orkut}, and \emph{Twitter}~\cite{kwak2010twitter} is considered as a massive dataset.
A short description of each dataset is given in Appendix \ref{appendix:dataset}.

\noindent \textbf{Setting of Dynamicity}.
We follow GloDyNE's setup for the real dataset \emph{Cora}~\cite{hou2020glodyne}, a widely accepted dataset whose dynamic graph has 11 snapshots.
For other datasets, following the setting in LocalAction~\cite{liu2019real} and DNE~\cite{du2018dynamic}, we start with a small number of nodes and gradually add the remaining ones to the graph individually (streaming scenario); for the discrete methods GloDyNE~\cite{hou2020glodyne}, dynGEM~\cite{goyal2018dyngem}, Dynnode2vec~\cite{mahdavi2018dynnode2vec}, the continuous streaming modifications are split into 100 discrete modifications.
To reflect the reality of the recent exponential rise in social network users, we utilize a considerably lower initial setting of 1000 nodes compared with the methods above.

\noindent \textbf{Time Limits}.
For small and large datasets, methods that have not produced results for more than $3$ days ($72$ hours) will not be included in the results. For massive datasets, methods that have not produced results for more than $7$ days ($168$ hours) will not be included in the results. 
 
\noindent \textbf{Paramater Settings}. 
The embedding dimension $d$ of all methods is set to $128$ for a fair comparison. 
For DAMF, the damping factor $\alpha$ for PPR is set to $0.3$ (except for the node classification task on \emph{Cora} where it is set to $0.03$) and the error tolerance $\epsilon$ is set to $10^{-5}$.


\subsection{Experimental Tasks}
\label{sec:tasks}
Dynamic network embedding is tested on three tasks: node classification, link prediction and graph reconstruction.
 
\textbf{Node Classification} is a common model training task to obtain the labels based on the embedding of each node by training a simple classifier.
Following previous work~\cite{tsitsulin2018verse, yang2020nrp}, we randomly select a subset of nodes to train a one-vs-all logistic regression classifier, and then use the classifier to predict the labels of the remaining nodes.
For each node $v$, we first concatenate the normalized context and content vectors as the feature representation of $v$ and then feed it to the classifier.

\textbf{Link Prediction} is the method of predicting the possibility of a link between two graph nodes.
Based on earlier research, we first generate the modified graph $\mathcal{G'}$ by removing 30\% of randomly selected edges from the input graph $\mathcal{G}$, and then construct embeddings on $\mathcal{G'}$.
Then, we build the testing set $\mathcal{E}_{test}$ by selecting the node pairs connected by the removed edges and an equal number of unconnected node pairs in $\mathcal{G}$. 

The inner product of node's embedding are used to make predictions.
Results are evaluated by \emph{Area Under Curve (AUC)} and \emph{Average Precision(AP)}.  

\textbf{Graph Reconstruction} is a basic objective of network embedding.
According to earlier work, we begin by selecting a set $\mathcal{S}$ of node pairs from the input graph $\mathcal{G}$. Afterward, we apply the same method used in link prediction to generate a score for each pair. Then, we calculate the \emph{precision@K}, the proportion of the top-K node pairings that match edges in $\mathcal{G}$.
 
\begin{table}[h]
\captionsetup{justification=centering}
\setlength{\abovecaptionskip}{0.2cm}
\caption{Link prediction results on small graphs\\
TLE: Time Limit Exceeded, $\times $: No Legal Output
}\
\setlength{\tabcolsep}{0.7mm}
\label{tab:lp}
\begin{tabular}{@{}c|cccccc@{}}
\toprule
\multirow{2}{*}{\diagbox{\textbf{Method}}{\textbf{Dataset}}} & \multicolumn{3}{c|}{\textbf{AUC}}                    & \multicolumn{3}{c}{\textbf{AP}}                              \\ \cmidrule(l){2-7} 
                           & Wiki            & Cora            &\multicolumn{1}{c|}{Flickr}          & Wiki            & Cora                     & Flickr          \\ \midrule
LocalAction                & 0.5083          & 0.5562          & \multicolumn{1}{c|}{0.4995}         & 0.5812          & 0.5923                   & 0.5833          \\
GloDyNE                    & 0.6386          & \textbf{0.9296}          & \multicolumn{1}{c|}{0.8511}               & 0.6844          & \textbf{0.9383}          & 0.8639               \\
Dynnode2vec                & 0.5904          & 0.8242          & \multicolumn{1}{c|}{0.8187}          & 0.6639          & 0.8975                   & 0.8476          \\
RandNE                     & 0.1092          & 0.8324          & \multicolumn{1}{c|}{TLE}               & 0.3208          & 0.8435                   & TLE               \\
DynGEM                 & $\times$               & 0.8340          & \multicolumn{1}{c|}{TLE}               & $\times$               & 0.8568                   & TLE               \\
DNE                        & 0.1845          & 0.8407          & \multicolumn{1}{c|}{0.5361}          & 0.3391          & 0.8662                   & 0.5287          \\ \midrule
DAMF($\alpha = 1$)                       & 0.6618          & 0.8555          & \multicolumn{1}{c|}{0.9269}          & 0.7541          & 0.8853                   & 0.9474          \\
DAMF         & \textbf{0.7543} & 0.9010 & \multicolumn{1}{c|}{\textbf{0.9550}} & \textbf{0.7840} & 0.9168 & \textbf{0.9615} \\ \bottomrule
\end{tabular}
\end{table}

\subsection{Experiment}
\textbf{Experiment on Small Graphs.}
\label{sec:smallg}
We conduct experiments on small graphs \emph{Wiki}, \emph{Cora} and \emph{Flickr} with all three tasks of Node Classification, Link Prediction, and Graph Reconstruction.

\begin{figure}[b]
\centering 
\setlength{\abovecaptionskip}{0.cm}
\includegraphics[height=3.1cm,width=8.5cm]{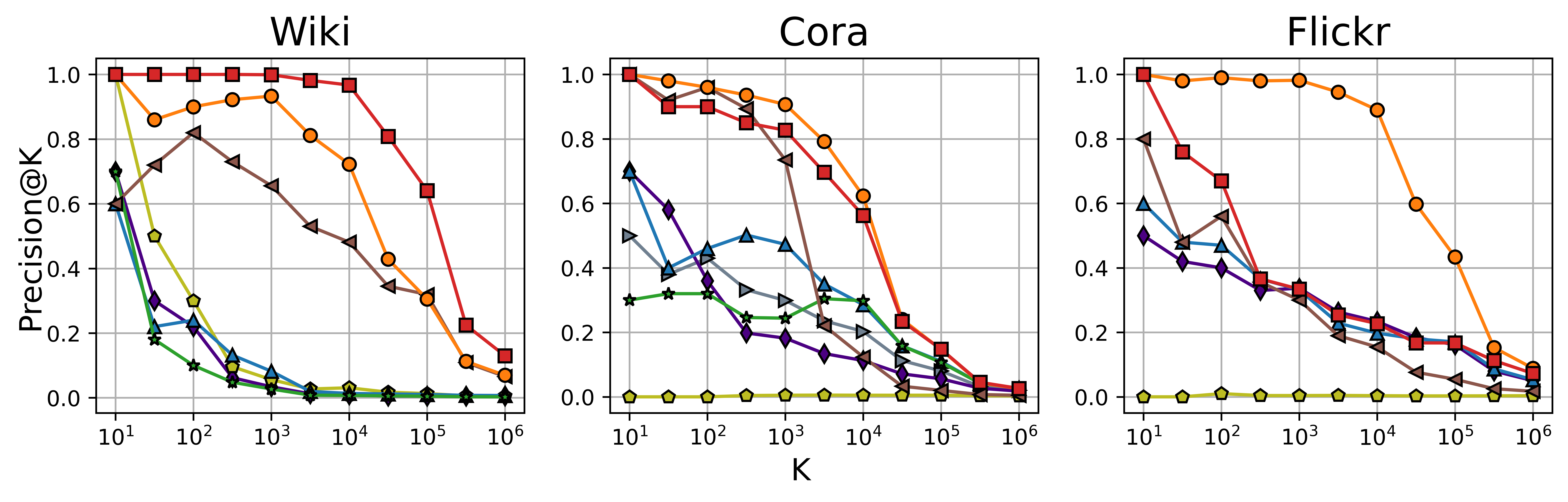}
\caption{Graph Reconstruction (\emph{precision@K})}
\label{fig:gr}
\end{figure}

The results of node classification, link prediction and graph reconstruction are displayed in Figure \ref{fig:nc}, Table \ref{tab:lp} and Figure \ref{fig:gr} respectively. 
On the node classification task, DAMF outperforms all other competitors significantly at \emph{Wiki}, is comparable to other methods on other datasets, and had better performance when the ratio of training data was smaller.
For link prediction, we can see that DAMF outperforms the baselines on almost all datasets based on both \emph{AUC} and \emph{AP} values, except on the \emph{Cora} dataset where it is slightly inferior compared to GloDyNE. In terms of the result of graph reconstruction, DAMF is significantly better than the other baselines. 
Overall, the DAMF method performs well and is stable across different data sets.


\begin{table}[h]
\setlength{\abovecaptionskip}{0.cm}
\caption{Link prediction results on large and massive graphs}
\label{tab:llp}
\setlength{\tabcolsep}{0.35mm}
\begin{tabular}{@{}c|cccccc@{}}
\toprule
\multirow{2}{*}{\diagbox{\textbf{Method}}{\textbf{Dataset}}} & \multicolumn{3}{c|}{\textbf{AUC}}                                                  & \multicolumn{3}{c}{\textbf{AP}}                              \\ \cmidrule(l){2-7} 
                        & Youtube         & Orkut           & \multicolumn{1}{c|}{Twitter}         & Youtube         & Orkut           & Twitter         \\ \midrule
LocalAction             & 0.4849          & 0.4978          & \multicolumn{1}{c|}{0.5006}          & 0.5548          & 0.5355          & 0.5923          \\ \midrule
DAMF($\alpha=1$)        & 0.7648          & 0.8662          & \multicolumn{1}{c|}{0.8732}                & 0.8290          & 0.8828          &    0.9018             \\
DAMF                    & \textbf{0.7946} & \textbf{0.8724} & \multicolumn{1}{c|}{\textbf{0.9055}} & \textbf{0.8510} & \textbf{0.8882} & \textbf{0.9353} \\ \bottomrule
\end{tabular}
\end{table}
 
\textbf{Experiment on Large Graphs.}
\label{sec:largeg}
Our extensive graph includes \emph{YouTube} and \emph{Orkut} datasets. Unlike small graphs, we only conduct Node Classification and Link Prediction. Due to the high number of potential $\binom{n}{2}$ pairs, we discard the Graph Reconstruction task. It is worth noting that only LocalAction reaches the set time limit among all competitors.

The results of node classification presented in Figure \ref{fig:nc} indicate that DAMF performs significantly better in comparison to LocalAction on both large graphs.
The \emph{AUC} and \emph{AP} results for the link prediction task are shown in Table \ref{tab:llp}. In both indicators, DAMF has a significant advantage over the only competitor that can complete the work within the allotted time on large graphs.

\textbf{Experiment on the Massive Graphs.}
\label{sec:massive}
Massive graphs with billion-level edges such as large-scale social networks are widely available in the industry.
However, to the best of our knowledge, dynamic network embedding studies on massive graphs with billion-level edges are unprecedented.
We conduct the first dynamic network embedding experiments on a billion-level edge dataset \emph{Twitter} with 41 million nodes and \textbf{1.4 billion} edges, and mapping each node as a $128$-dimensional vector which is more than 10 times as many learnable parameters as BERT-Large~\cite{kenton2019bert}.

We conduct a link prediction task to determine the nodes' connection probability. Notably, we reduce the ratio of edges deleted in the first stage from 30\% to 0.5\% to ensure graph connectivity when there are only a few nodes.
We discard node classification and graph reconstruction tasks since the Twitter dataset lacks labels to build a classifier and the excessive number of potential pairings makes reconstruction unattainable.

Table \ref{tab:llp} shows the \emph{AUC} and \emph{AP} values. Overall, DAMF performs well on both evaluation indicators, obtaining an \emph{AUC} value of $0.9055$ and an \emph{AP} of $0.9353$. Moreover, DAMF's overall updating time of $\textbf{110}$ hours demonstrates that DAMF is capable of changing parameters at the billion level in under $10$ milliseconds.
 
All baseline methods except LocalAction~\cite{liu2019real} exceed the set time limit (7 days) and are therefore excluded. However, LocalAction seems to fail to converge on this dataset since its \emph{AUC} on link prediction is only $0.5006$.
 
\subsection{Efficency Study}
\label{sec:runtime}
\begin{table}[t]
\setlength{\abovecaptionskip}{0.cm}
\caption{Running time}
\setlength{\abovecaptionskip}{0.cm}
\setlength{\tabcolsep}{0.8mm}
\label{tab:run}
\begin{tabular}{@{}c|ccc|cc@{}}
\toprule
\multirow{2}{*}{\textbf{\diagbox{Method}{Size}}} & \multicolumn{3}{c|}{\textbf{Small}}                                       & \multicolumn{2}{c}{\textbf{Large}}                 \\ \cmidrule(l){2-6} 
                               & Wiki & Cora           & Flickr              & Youtube             & Orkut               \\ \midrule
LocalAction                    & \textbf{7s}              &  \textbf{6s}      &   \textbf{5m28s}    &   \textbf{5m10s} & \textbf{1h43m}       \\ 
GloDyNE                        & \textbf{10m45s}           & \textbf{1m5s}  & \textbf{18h38m}     & \textgreater{}3days & \textgreater{}3days \\ 
Dynnode2vec                    & \textbf{1m58s}            & \textbf{1m40s} & \textbf{5h10m}      & \textgreater{}3days & \textgreater{}3days \\ 
RandNE                         & \textbf{3m}               &  \textbf{5s} & \textgreater{}3days & \textgreater{}3days & \textgreater{}3days \\ 
DynGEM                     & \textbf{17h4m}            & \textbf{6h20m} & \textgreater{}3days & \textgreater{}3days & \textgreater{}3days \\ 
DNE                            & \textbf{22m38s}           & \textbf{8m57s}  & \textbf{8h11m}      & \textgreater{}3days & \textgreater{}3days \\ \midrule
DAMF($\alpha=1$)               & \textbf{36s}              & \textbf{1m19s} & \textbf{33m13s}     & \textbf{3h10m}      & \textbf{8h15m}      \\ 
DAMF                           & \textbf{2m47s}            & \textbf{2m6s}  & \textbf{1h33m}      & \textbf{3h39m}      & \textbf{14h7m}      \\ 
\bottomrule
\end{tabular}
\end{table}

Table \ref{tab:run} shows the running time of each method on small and large datasets. The experimental results show that the speed of DAMF increases more consistently as the size of the dataset expands. According to Table \ref{tab:lp} and Table \ref{tab:llp}, although LocalAction is faster, its \emph{AUC} is only slightly greater than 0.5, indicating that LocalAction makes unbearable tread-offs to obtain speed, whereas DAMF maintains stable effectiveness.

\subsection{Ablation Study}
\label{sec:ablation}

To investigate whether the dynamic embedding enhancement better captures the information of higher-order neighbors and thus improves the quality of dynamic network embedding, we obtain unenhanced embeddings by setting the PageRank damping factor $\alpha$ to 1.
The experimental results in Figure \ref{fig:nc} ,Table \ref{tab:lp}, Table \ref{tab:llp} and Figure \ref{fig:gr} show that the enhanced embeddings are significantly better in node classification, link prediction and graph reconstruction on \emph{Wiki}, demonstrating the effectiveness of the dynamic embedding enhancement.

%% file: 06Conclusion.tex
\section{Conclusion}
In this work, we propose the Dynamic Adjacency Matrix Factorization (DAMF) algorithm, which utilizes a projection onto the embedding space and modifies a few node embeddings to achieve an efficient adjacency matrix decomposition method.
In addition, we use dynamic Personalized PageRank to enhance the embedding to capture high-order neighbors' information dynamically. 
Experimental results show that our proposed method performs well in node classification, link prediction, and graph reconstruction, achieving an average dynamic network embedding update time of 10ms on billion-edge graphs.

\noindent \textbf{Acknowledgement.} This work is supported by NSFC (No.62176233), 
National Key Research and Development Project of China 

\noindent (No.2018AAA0101900) and Fundamental Research Funds for the Central Universities.

%% file: 07Appendix.tex
\section*{Appendix}

\section{Initialization of DAMF}

\begin{algorithm}[h] 
    \caption{Initialization of DAMF}
    \label{algo:init}
    \LinesNumbered 
    \KwIn{A graph $\mathcal{G}$ with adjacency matrix $\mathbf{A}$, embedding dimension $d$, PageRank damping factor $\alpha$, error tolerance $\epsilon$.}
    \KwOut{$\mathbf{X}_b, \mathbf{Y}_b, \mathbf{Z}_b, \mathbf{P_X}, \mathbf{P_Y}, \mathbf{r}$}
    $\mathbf{U}, \mathbf{\Sigma}, \mathbf{V} \gets \textbf{\texttt{t-SVD}}(\mathbf{A}, d)$\;
    $\mathbf{X}_b \gets \mathbf{U} \mathbf{\Sigma}^{1/2}, \quad \mathbf{Y}_b \gets \mathbf{V} \mathbf{\Sigma}^{1/2}$\;
    $\mathbf{P_X}\gets \mathbf{I},\quad \mathbf{P_Y} \gets \mathbf{I}$\;
    $\mathbf{r} \gets \mathbf{X}_b, \quad \mathbf{Z}_b \gets \mathbf{O}$\;
    $\mathbf{Z}_b, \mathbf{r} \gets \texttt{Propagation}(\mathbf{Z}_b, \mathbf{r}, \mathcal{G}, \alpha, \epsilon)$\;
    \KwRet{$\mathbf{X}_b, \mathbf{Y}_b, \mathbf{Z}_b, \mathbf{P_X}, \mathbf{P_Y}, \mathbf{r}$}\;
\end{algorithm}

Algorithm \ref{algo:init} gives a detailed pseudo-code for the initialization of the DAMF. The random projection-based truncated SVD algorithm \cite{halko2011randomsvd} is able to complete the t-SVD in $O(nd^2+md)$ time, while the time complexity of the PPR enhancement that follows the initialization is O($\frac{ncd}{\alpha \epsilon}$)~\cite{zheng2022instant}. Overall, the time complexity of the initialization step is $O(nd^2 + md + \frac{ncd}{\alpha \epsilon})$.

\section{Proof}
\subsection{Proof of Lemma \ref{lemma:nzr}}
\begin{proof}[Proof of Lemma \ref{lemma:nzr}]
The $i$-th row of the result matrix of the matrix multiplication of $\mathbf{BC}$ can be considered as the $i$-th row of $\mathbf{B}$ multiplied by the matrix $\mathbf{C}$. 
Therefore, if the $i$-th row of $\mathbf{B}$ is all-zero, the $i$-th row of the result matrix will also be all-zero. 
Since $\mathbf{B}$ has only t non-zero rows, $\mathbf{BC}$ has at most $t$ non-zero rows.
\end{proof}

\subsection{Proof of Proposition \ref{proposition:fast}}
\begin{proof}[Proof of Proposition \ref{proposition:fast}]
\begin{equation}
\begin{aligned}
\left \|  (\mathbf{I}-\mathbf{UU}^\top)\vec{x} \right \| _{2}
&=\sqrt{((\mathbf{I}-\mathbf{UU}^\top)\vec{x})^\top (\mathbf{I}-\mathbf{UU}^\top)\vec{x}} \\
&=\sqrt{\vec{x}^\top (\mathbf{I}-\mathbf{UU}^\top)^\top (\mathbf{I}-\mathbf{UU}^\top) \vec{x}} \\
&=\sqrt{\vec{x}^\top (\mathbf{I}-2\mathbf{UU}^\top + \mathbf{UU}^\top\mathbf{UU}^\top) \vec{x}} \\
&=\sqrt{\vec{x}^\top (\mathbf{I}-\mathbf{UU}^\top) \vec{x}} \\
&=\sqrt{\vec{x}^\top \vec{x} - \vec{x}^\top \mathbf{UU}^\top \vec{x}} \\
&=\sqrt{\Vert \vec{x} \Vert_2^2 - \Vert \mathbf{U}^\top\vec{x} \Vert_2^2} \\
\end{aligned}
\end{equation}
\end{proof}

\subsection{Proof of Lemma \ref{lemma:complexity_A}}
\label{appendix:complexity_A}
\begin{proof}

When performing the matrix multiplication of $\mathbf{B}^\top \mathbf{A}$, the element on the $i$-th row and $j$-th column of the result matrix can be obtained by the product of the $i$-th row of $\mathbf{B}^\top$ and the $j$-th row of $\mathbf{A}$ with
\begin{equation}
    (\mathbf{B}^\top \mathbf{A})[i, j] = \mathbf{B}^\top [i] \cdot \mathbf{A}[:,j]
\end{equation}

Because $\mathbf{B}$ has at most $t$ non-zero rows, $\mathbf{B}^\top[i]$ has at most $t$ non-zero elements. 
By skipping zeros in calculating the product, the above equation can be calculated with the time complexity of $O(t)$.
And since $\mathbf{B}^\top \mathbf{A}$ is a $q$-by-$p$ matrix, the time complexity of calculating $\mathbf{B}^\top \mathbf{A}$ is $O(tpq)$.

Algorithm \ref{algo:complexity_A} is a pseudo-code for the above scheme.

\end{proof}

\begin{algorithm}[hb] 
    \caption{Algorithm for Lemma \ref{lemma:complexity_A}}
    \label{algo:complexity_A}
    \LinesNumbered 
    \KwIn{$\mathbf{A}\in \mathbb{R}^{n\times p}, \mathbf{B} \in \mathbb{R}^{n \times q}$ and $\mathbf{B}$ has $t$ non-zero rows.}
    \KwOut{$\mathbf{B}^\top \mathbf{A}$}
    $\mathbf{C} \gets \mathbf{O}_{q\times p}$\;
    \ForEach{l with $\mathbf{B}[l]$ is non-zero}{
        \For{$i\gets 1$ to $q$}{
            \For{$j\gets 1$ to $p$}{
                $\mathbf{C}[i,j] \gets \mathbf{C}[i,j] + \mathbf{A}[l,j] \times \mathbf{B}[l, i]$\;
            }
        }
    }
    \KwRet{$\mathbf{C}$}
\end{algorithm}

\subsection{Proof of Lemma \ref{lemma:complexity_B}}
\label{appendix:complexity_B}
\begin{proof}
From Lemma \ref{lemma:nzr}, there are only at most $t$ non-zero rows in the result of $\mathbf{BC}$. 
So, by skipping the calculation of all-zero rows, the time complexity of of calculating $\mathbf{BC}$ is $O(tpq)$.

Algorithm \ref{algo:complexity_B} is a pseudo-code for the above scheme.
\end{proof}
\begin{algorithm}[hb] 
    \caption{Algorithm for Lemma \ref{lemma:complexity_B}}
    \label{algo:complexity_B}
    \LinesNumbered 
    \KwIn{$\mathbf{B} \in \mathbb{R}^{n \times q}$ and $\mathbf{B}$ has $t$ non-zero rows, $\mathbf{C} \in \mathbb{R}^{q \times p}$}
    \KwOut{$\mathbf{BC}$}
    $\mathbf{D} \gets \mathbf{O}_{n\times p}$\;
    \ForEach{l with $\mathbf{B}[l]$ is non-zero}{
        \For{$i\gets 1$ to $q$}{
            \For{$j\gets 1$ to $p$}{
                $\mathbf{D}[l,j] \gets \mathbf{D}[l,j] + \mathbf{B}[l,i] \times \mathbf{C}[i, j]$\;
            }
        }
    }
    \KwRet{$\mathbf{D}$}
\end{algorithm}

\section{Detailed Experimental Settings}

All experiments are conducted using 8 threads on a Linux machine powered by an AMD Epyc 7H12@3.2GHz and 768GB RAM.
For baselines that require a GPU, we used an additional NVIDIA GeForce RTX 3090 with 24G memory.
The experimental results for each task are the average of the results of five experiments.

\subsection{Additional Details of Link Prediction}

The objective of link prediction is, for the directed node pair $(u, v)$, to predict whether there is a directed edge from $u$ to $v$.
For each pair of nodes $(u, v)$ in the test set, we determine a score by taking the inner product of $u$'s context vector and $v$'s content vector.
For methods that do not distinguish between context embedding and content embedding, we set both their context embedding and content embedding to be node embeddings.
For undirected graphs, we calculate link prediction scores for both directions separately, and then select the larger value as the score.

\subsection{Additional Details of Graph Reconstruction}
For \emph{Wiki} and \emph{Cora} datasets, we define set $\mathcal{S}$ as the collection of all conceivable node pairs. On \emph{Flickr}, we construct $\mathcal{S}$ by taking a 1\% sample of all possible pairs of nodes. Figure \ref{fig:gr} depicts the performance of all methods in the graph reconstruction task with $K$ values ranging from $10^{1}$ to $10^{6}$. DAMF performs better than its competitors on all the datasets and for nearly every $K$ value. The remarkable accuracy of DAMF is especially noticeable on the \emph{Wiki} and \emph{Cora} datasets as $K$ increases.

\subsection{Dataset}

\label{appendix:dataset}
\noindent 

\textbf{Wiki~\cite{wiki}} is a hyperlinked network of Wikipedia. Each node in the directed graph represents a page, and the edges represent hyperlinks.

\textbf{Cora~\cite{hou2020glodyne}} is a dynamic citation undirected network where each node represents a paper, the edges represent citations, and each article is labeled with the domain to which the article belongs.

\textbf{Flickr~\cite{flickr}} is an undirected network of user links on Flickr, where nodes represent users and labels are groups of interest to the user.

\textbf{YouTube~\cite{youtube}} is a video-sharing website where users can upload and watch videos, and share, comment and rate them. Each user is labeled with their favorite video genre.

\textbf{Orkut~\cite{orkut}} is an undirected social network where nodes represent users and edges represent user interactions. Users are labeled with the community they are in.

\textbf{Twitter~\cite{kwak2010twitter}} is a massive-scale directed social network where nodes represent users and edges represent the following relationship. 

\subsection{Code}
We use the following code as our baseline:

Dynnode2vec~\cite{mahdavi2018dynnode2vec}: https://github.com/pedugnat/dynnode2vec

GloDyNE~\cite{hou2020glodyne}: https://github.com/houchengbin/GloDyNE

DynGEM~\cite{goyal2018dyngem}: https://github.com/palash1992/DynamicGEM

DNE~\cite{du2018dynamic}: https://github.com/lundu28/DynamicNetworkEmbedding

RandNE~\cite{zhang2018billion}: https://github.com/ZW-ZHANG/RandNE

As we could not find the code for LocalAction~\cite{liu2019real}, we re-implemented the pseudo-code from the paper in Python and made our implementation available along with our experimental code.

Our code for experiments is available on:

https://github.com/zjunet/DAMF


%% file: main.bbl

\begin{thebibliography}{54}


\ifx \showCODEN    \undefined \def \showCODEN     #1{\unskip}     \fi
\ifx \showDOI      \undefined \def \showDOI       #1{#1}\fi
\ifx \showISBNx    \undefined \def \showISBNx     #1{\unskip}     \fi
\ifx \showISBNxiii \undefined \def \showISBNxiii  #1{\unskip}     \fi
\ifx \showISSN     \undefined \def \showISSN      #1{\unskip}     \fi
\ifx \showLCCN     \undefined \def \showLCCN      #1{\unskip}     \fi
\ifx \shownote     \undefined \def \shownote      #1{#1}          \fi
\ifx \showarticletitle \undefined \def \showarticletitle #1{#1}   \fi
\ifx \showURL      \undefined \def \showURL       {\relax}        \fi
\providecommand\bibfield[2]{#2}
\providecommand\bibinfo[2]{#2}
\providecommand\natexlab[1]{#1}
\providecommand\showeprint[2][]{arXiv:#2}

\bibitem[Abu-El-Haija et~al\mbox{.}(2018)]%
        {abu2018watch}
\bibfield{author}{\bibinfo{person}{Sami Abu-El-Haija}, \bibinfo{person}{Bryan
  Perozzi}, \bibinfo{person}{Rami Al-Rfou}, {and} \bibinfo{person}{Alexander~A
  Alemi}.} \bibinfo{year}{2018}\natexlab{}.
\newblock \showarticletitle{Watch your step: Learning node embeddings via graph
  attention}.
\newblock \bibinfo{journal}{\emph{Advances in neural information processing
  systems}}  \bibinfo{volume}{31} (\bibinfo{year}{2018}).
\newblock


\bibitem[Arsov and Mirceva(2019)]%
        {Arsov2019NetworkEA}
\bibfield{author}{\bibinfo{person}{Nino Arsov} {and} \bibinfo{person}{Georgina
  Mirceva}.} \bibinfo{year}{2019}\natexlab{}.
\newblock \showarticletitle{Network Embedding: An Overview}.
\newblock \bibinfo{journal}{\emph{ArXiv}}  \bibinfo{volume}{abs/1911.11726}
  (\bibinfo{year}{2019}).
\newblock


\bibitem[Bojchevski et~al\mbox{.}(2020)]%
        {bojchevski2020scaling}
\bibfield{author}{\bibinfo{person}{Aleksandar Bojchevski},
  \bibinfo{person}{Johannes Klicpera}, \bibinfo{person}{Bryan Perozzi},
  \bibinfo{person}{Amol Kapoor}, \bibinfo{person}{Martin Blais},
  \bibinfo{person}{Benedek R{\'o}zemberczki}, \bibinfo{person}{Michal Lukasik},
  {and} \bibinfo{person}{Stephan G{\"u}nnemann}.}
  \bibinfo{year}{2020}\natexlab{}.
\newblock \showarticletitle{Scaling graph neural networks with approximate
  pagerank}. In \bibinfo{booktitle}{\emph{Proceedings of the 26th ACM SIGKDD
  International Conference on Knowledge Discovery \& Data Mining}}.
  \bibinfo{pages}{2464--2473}.
\newblock


\bibitem[Cao et~al\mbox{.}(2016)]%
        {cao2016deep}
\bibfield{author}{\bibinfo{person}{Shaosheng Cao}, \bibinfo{person}{Wei Lu},
  {and} \bibinfo{person}{Qiongkai Xu}.} \bibinfo{year}{2016}\natexlab{}.
\newblock \showarticletitle{Deep neural networks for learning graph
  representations}. In \bibinfo{booktitle}{\emph{Proceedings of the AAAI
  conference on artificial intelligence}}, Vol.~\bibinfo{volume}{30}.
\newblock


\bibitem[Chen and Tong(2015)]%
        {chen2015fast}
\bibfield{author}{\bibinfo{person}{Chen Chen} {and} \bibinfo{person}{Hanghang
  Tong}.} \bibinfo{year}{2015}\natexlab{}.
\newblock \showarticletitle{Fast eigen-functions tracking on dynamic graphs}.
  In \bibinfo{booktitle}{\emph{Proceedings of the 2015 SIAM international
  conference on data mining}}. SIAM, \bibinfo{pages}{559--567}.
\newblock


\bibitem[Choudhary et~al\mbox{.}(2021)]%
        {Choudhary2021ASO}
\bibfield{author}{\bibinfo{person}{Shivani Choudhary}, \bibinfo{person}{Tarun
  Luthra}, \bibinfo{person}{Ashima Mittal}, {and} \bibinfo{person}{Rajat
  Singh}.} \bibinfo{year}{2021}\natexlab{}.
\newblock \showarticletitle{A Survey of Knowledge Graph Embedding and Their
  Applications}.
\newblock \bibinfo{journal}{\emph{ArXiv}}  \bibinfo{volume}{abs/2107.07842}
  (\bibinfo{year}{2021}).
\newblock


\bibitem[Cui et~al\mbox{.}(2018)]%
        {cui2018survey}
\bibfield{author}{\bibinfo{person}{Peng Cui}, \bibinfo{person}{Xiao Wang},
  \bibinfo{person}{Jian Pei}, {and} \bibinfo{person}{Wenwu Zhu}.}
  \bibinfo{year}{2018}\natexlab{}.
\newblock \showarticletitle{A survey on network embedding}.
\newblock \bibinfo{journal}{\emph{IEEE transactions on knowledge and data
  engineering}} \bibinfo{volume}{31}, \bibinfo{number}{5}
  (\bibinfo{year}{2018}), \bibinfo{pages}{833--852}.
\newblock


\bibitem[Du et~al\mbox{.}(2018)]%
        {du2018dynamic}
\bibfield{author}{\bibinfo{person}{Lun Du}, \bibinfo{person}{Yun Wang},
  \bibinfo{person}{Guojie Song}, \bibinfo{person}{Zhicong Lu}, {and}
  \bibinfo{person}{Junshan Wang}.} \bibinfo{year}{2018}\natexlab{}.
\newblock \showarticletitle{Dynamic network embedding: An extended approach for
  skip-gram based network embedding.}. In \bibinfo{booktitle}{\emph{IJCAI}},
  Vol.~\bibinfo{volume}{2018}. \bibinfo{pages}{2086--2092}.
\newblock


\bibitem[Gasteiger et~al\mbox{.}(2018)]%
        {gasteiger2018predict}
\bibfield{author}{\bibinfo{person}{Johannes Gasteiger},
  \bibinfo{person}{Aleksandar Bojchevski}, {and} \bibinfo{person}{Stephan
  G{\"u}nnemann}.} \bibinfo{year}{2018}\natexlab{}.
\newblock \showarticletitle{Predict then propagate: Graph neural networks meet
  personalized pagerank}.
\newblock \bibinfo{journal}{\emph{arXiv preprint arXiv:1810.05997}}
  (\bibinfo{year}{2018}).
\newblock


\bibitem[Gong et~al\mbox{.}(2020)]%
        {gong2020ctine}
\bibfield{author}{\bibinfo{person}{Maoguo Gong}, \bibinfo{person}{Shunfei Ji},
  \bibinfo{person}{Yu Xie}, \bibinfo{person}{Yuan Gao}, {and}
  \bibinfo{person}{AK Qin}.} \bibinfo{year}{2020}\natexlab{}.
\newblock \showarticletitle{Exploring temporal information for dynamic network
  embedding}.
\newblock \bibinfo{journal}{\emph{IEEE Transactions on Knowledge and Data
  Engineering}} \bibinfo{volume}{34}, \bibinfo{number}{8}
  (\bibinfo{year}{2020}), \bibinfo{pages}{3754--3764}.
\newblock


\bibitem[Goyal et~al\mbox{.}(2018)]%
        {goyal2018dyngem}
\bibfield{author}{\bibinfo{person}{Palash Goyal}, \bibinfo{person}{Nitin
  Kamra}, \bibinfo{person}{Xinran He}, {and} \bibinfo{person}{Yan Liu}.}
  \bibinfo{year}{2018}\natexlab{}.
\newblock \showarticletitle{Dyngem: Deep embedding method for dynamic graphs}.
\newblock \bibinfo{journal}{\emph{arXiv preprint arXiv:1805.11273}}
  (\bibinfo{year}{2018}).
\newblock


\bibitem[Grover and Leskovec(2016)]%
        {grover2016node2vec}
\bibfield{author}{\bibinfo{person}{Aditya Grover} {and} \bibinfo{person}{Jure
  Leskovec}.} \bibinfo{year}{2016}\natexlab{}.
\newblock \showarticletitle{node2vec: Scalable feature learning for networks}.
  In \bibinfo{booktitle}{\emph{Proceedings of the 22nd ACM SIGKDD international
  conference on Knowledge discovery and data mining}}.
  \bibinfo{pages}{855--864}.
\newblock


\bibitem[Halko et~al\mbox{.}(2011)]%
        {halko2011randomsvd}
\bibfield{author}{\bibinfo{person}{Nathan Halko}, \bibinfo{person}{Per-Gunnar
  Martinsson}, {and} \bibinfo{person}{Joel~A Tropp}.}
  \bibinfo{year}{2011}\natexlab{}.
\newblock \showarticletitle{Finding structure with randomness: Probabilistic
  algorithms for constructing approximate matrix decompositions}.
\newblock \bibinfo{journal}{\emph{SIAM review}} \bibinfo{volume}{53},
  \bibinfo{number}{2} (\bibinfo{year}{2011}), \bibinfo{pages}{217--288}.
\newblock


\bibitem[Hamilton et~al\mbox{.}(2017)]%
        {hamilton2017representation}
\bibfield{author}{\bibinfo{person}{William~L Hamilton}, \bibinfo{person}{Rex
  Ying}, {and} \bibinfo{person}{Jure Leskovec}.}
  \bibinfo{year}{2017}\natexlab{}.
\newblock \showarticletitle{Representation learning on graphs: Methods and
  applications}.
\newblock \bibinfo{journal}{\emph{arXiv preprint arXiv:1709.05584}}
  (\bibinfo{year}{2017}).
\newblock


\bibitem[Hou et~al\mbox{.}(2020)]%
        {hou2020glodyne}
\bibfield{author}{\bibinfo{person}{Chengbin Hou}, \bibinfo{person}{Han Zhang},
  \bibinfo{person}{Shan He}, {and} \bibinfo{person}{Ke Tang}.}
  \bibinfo{year}{2020}\natexlab{}.
\newblock \showarticletitle{Glodyne: Global topology preserving dynamic network
  embedding}.
\newblock \bibinfo{journal}{\emph{IEEE Transactions on Knowledge and Data
  Engineering}} \bibinfo{volume}{34}, \bibinfo{number}{10}
  (\bibinfo{year}{2020}), \bibinfo{pages}{4826--4837}.
\newblock


\bibitem[Kenton and Toutanova(2019)]%
        {kenton2019bert}
\bibfield{author}{\bibinfo{person}{Jacob Devlin Ming-Wei~Chang Kenton} {and}
  \bibinfo{person}{Lee~Kristina Toutanova}.} \bibinfo{year}{2019}\natexlab{}.
\newblock \showarticletitle{BERT: Pre-training of Deep Bidirectional
  Transformers for Language Understanding}. In
  \bibinfo{booktitle}{\emph{Proceedings of NAACL-HLT}}.
  \bibinfo{pages}{4171--4186}.
\newblock


\bibitem[Kwak et~al\mbox{.}(2010)]%
        {kwak2010twitter}
\bibfield{author}{\bibinfo{person}{Haewoon Kwak}, \bibinfo{person}{Changhyun
  Lee}, \bibinfo{person}{Hosung Park}, {and} \bibinfo{person}{Sue Moon}.}
  \bibinfo{year}{2010}\natexlab{}.
\newblock \showarticletitle{What is Twitter, a social network or a news
  media?}. In \bibinfo{booktitle}{\emph{Proceedings of the 19th international
  conference on World wide web}}. \bibinfo{pages}{591--600}.
\newblock


\bibitem[Li et~al\mbox{.}(2017)]%
        {li2017dane}
\bibfield{author}{\bibinfo{person}{Jundong Li}, \bibinfo{person}{Harsh Dani},
  \bibinfo{person}{Xia Hu}, \bibinfo{person}{Jiliang Tang}, \bibinfo{person}{Yi
  Chang}, {and} \bibinfo{person}{Huan Liu}.} \bibinfo{year}{2017}\natexlab{}.
\newblock \showarticletitle{Attributed network embedding for learning in a
  dynamic environment}. In \bibinfo{booktitle}{\emph{Proceedings of the 2017
  ACM on Conference on Information and Knowledge Management}}.
  \bibinfo{pages}{387--396}.
\newblock


\bibitem[Liu et~al\mbox{.}(2019)]%
        {liu2019real}
\bibfield{author}{\bibinfo{person}{Xi Liu}, \bibinfo{person}{Ping-Chun Hsieh},
  \bibinfo{person}{Nick Duffield}, \bibinfo{person}{Rui Chen},
  \bibinfo{person}{Muhe Xie}, {and} \bibinfo{person}{Xidao Wen}.}
  \bibinfo{year}{2019}\natexlab{}.
\newblock \showarticletitle{Real-time streaming graph embedding through local
  actions}. In \bibinfo{booktitle}{\emph{Companion Proceedings of The 2019
  World Wide Web Conference}}. \bibinfo{pages}{285--293}.
\newblock


\bibitem[Ma et~al\mbox{.}(2018)]%
        {ma2018depthlgp}
\bibfield{author}{\bibinfo{person}{Jianxin Ma}, \bibinfo{person}{Peng Cui},
  {and} \bibinfo{person}{Wenwu Zhu}.} \bibinfo{year}{2018}\natexlab{}.
\newblock \showarticletitle{Depthlgp: Learning embeddings of out-of-sample
  nodes in dynamic networks}. In \bibinfo{booktitle}{\emph{Proceedings of the
  AAAI Conference on Artificial Intelligence}}, Vol.~\bibinfo{volume}{32}.
\newblock


\bibitem[Mahdavi et~al\mbox{.}(2018)]%
        {mahdavi2018dynnode2vec}
\bibfield{author}{\bibinfo{person}{Sedigheh Mahdavi}, \bibinfo{person}{Shima
  Khoshraftar}, {and} \bibinfo{person}{Aijun An}.}
  \bibinfo{year}{2018}\natexlab{}.
\newblock \showarticletitle{dynnode2vec: Scalable dynamic network embedding}.
  In \bibinfo{booktitle}{\emph{2018 IEEE International Conference on Big Data
  (Big Data)}}. IEEE, \bibinfo{pages}{3762--3765}.
\newblock


\bibitem[Mahoney(2011)]%
        {wiki}
\bibfield{author}{\bibinfo{person}{Matt Mahoney}.}
  \bibinfo{year}{2011}\natexlab{}.
\newblock \bibinfo{title}{Large text compression benchmark}.
\newblock
\newblock


\bibitem[Mislove et~al\mbox{.}(2007)]%
        {orkut}
\bibfield{author}{\bibinfo{person}{Alan Mislove}, \bibinfo{person}{Massimiliano
  Marcon}, \bibinfo{person}{Krishna~P Gummadi}, \bibinfo{person}{Peter
  Druschel}, {and} \bibinfo{person}{Bobby Bhattacharjee}.}
  \bibinfo{year}{2007}\natexlab{}.
\newblock \showarticletitle{Measurement and analysis of online social
  networks}. In \bibinfo{booktitle}{\emph{Proceedings of the 7th ACM SIGCOMM
  conference on Internet measurement}}. \bibinfo{pages}{29--42}.
\newblock


\bibitem[Nelson et~al\mbox{.}(2019)]%
        {nelson2019embed}
\bibfield{author}{\bibinfo{person}{Walter Nelson}, \bibinfo{person}{Marinka
  Zitnik}, \bibinfo{person}{Bo Wang}, \bibinfo{person}{Jure Leskovec},
  \bibinfo{person}{Anna Goldenberg}, {and} \bibinfo{person}{Roded Sharan}.}
  \bibinfo{year}{2019}\natexlab{}.
\newblock \showarticletitle{To embed or not: network embedding as a paradigm in
  computational biology}.
\newblock \bibinfo{journal}{\emph{Frontiers in genetics}}  \bibinfo{volume}{10}
  (\bibinfo{year}{2019}), \bibinfo{pages}{381}.
\newblock


\bibitem[Nguyen et~al\mbox{.}(2018)]%
        {nguyen2018ctdne}
\bibfield{author}{\bibinfo{person}{Giang~Hoang Nguyen},
  \bibinfo{person}{John~Boaz Lee}, \bibinfo{person}{Ryan~A Rossi},
  \bibinfo{person}{Nesreen~K Ahmed}, \bibinfo{person}{Eunyee Koh}, {and}
  \bibinfo{person}{Sungchul Kim}.} \bibinfo{year}{2018}\natexlab{}.
\newblock \showarticletitle{Continuous-time dynamic network embeddings}. In
  \bibinfo{booktitle}{\emph{Companion proceedings of the the web conference
  2018}}. \bibinfo{pages}{969--976}.
\newblock


\bibitem[Page et~al\mbox{.}(1999)]%
        {page1999pagerank}
\bibfield{author}{\bibinfo{person}{Lawrence Page}, \bibinfo{person}{Sergey
  Brin}, \bibinfo{person}{Rajeev Motwani}, {and} \bibinfo{person}{Terry
  Winograd}.} \bibinfo{year}{1999}\natexlab{}.
\newblock \bibinfo{booktitle}{\emph{The PageRank citation ranking: Bringing
  order to the web.}}
\newblock \bibinfo{type}{{T}echnical {R}eport}. \bibinfo{institution}{Stanford
  InfoLab}.
\newblock


\bibitem[Pareja et~al\mbox{.}(2020)]%
        {pareja2020evolvegcn}
\bibfield{author}{\bibinfo{person}{Aldo Pareja}, \bibinfo{person}{Giacomo
  Domeniconi}, \bibinfo{person}{Jie Chen}, \bibinfo{person}{Tengfei Ma},
  \bibinfo{person}{Toyotaro Suzumura}, \bibinfo{person}{Hiroki Kanezashi},
  \bibinfo{person}{Tim Kaler}, \bibinfo{person}{Tao Schardl}, {and}
  \bibinfo{person}{Charles Leiserson}.} \bibinfo{year}{2020}\natexlab{}.
\newblock \showarticletitle{Evolvegcn: Evolving graph convolutional networks
  for dynamic graphs}. In \bibinfo{booktitle}{\emph{Proceedings of the AAAI
  conference on artificial intelligence}}, Vol.~\bibinfo{volume}{34}.
  \bibinfo{pages}{5363--5370}.
\newblock


\bibitem[Perozzi et~al\mbox{.}(2014)]%
        {perozzi2014deepwalk}
\bibfield{author}{\bibinfo{person}{Bryan Perozzi}, \bibinfo{person}{Rami
  Al-Rfou}, {and} \bibinfo{person}{Steven Skiena}.}
  \bibinfo{year}{2014}\natexlab{}.
\newblock \showarticletitle{Deepwalk: Online learning of social
  representations}. In \bibinfo{booktitle}{\emph{Proceedings of the 20th ACM
  SIGKDD international conference on Knowledge discovery and data mining}}.
  \bibinfo{pages}{701--710}.
\newblock


\bibitem[Qiu et~al\mbox{.}(2021)]%
        {qiu2021lightne}
\bibfield{author}{\bibinfo{person}{Jiezhong Qiu}, \bibinfo{person}{Laxman
  Dhulipala}, \bibinfo{person}{Jie Tang}, \bibinfo{person}{Richard Peng}, {and}
  \bibinfo{person}{Chi Wang}.} \bibinfo{year}{2021}\natexlab{}.
\newblock \showarticletitle{Lightne: A lightweight graph processing system for
  network embedding}. In \bibinfo{booktitle}{\emph{Proceedings of the 2021
  international conference on management of data}}.
  \bibinfo{pages}{2281--2289}.
\newblock


\bibitem[Qiu et~al\mbox{.}(2019)]%
        {qiu2019netsmf}
\bibfield{author}{\bibinfo{person}{Jiezhong Qiu}, \bibinfo{person}{Yuxiao
  Dong}, \bibinfo{person}{Hao Ma}, \bibinfo{person}{Jian Li},
  \bibinfo{person}{Chi Wang}, \bibinfo{person}{Kuansan Wang}, {and}
  \bibinfo{person}{Jie Tang}.} \bibinfo{year}{2019}\natexlab{}.
\newblock \showarticletitle{Netsmf: Large-scale network embedding as sparse
  matrix factorization}. In \bibinfo{booktitle}{\emph{The World Wide Web
  Conference}}. \bibinfo{pages}{1509--1520}.
\newblock


\bibitem[Qiu et~al\mbox{.}(2018)]%
        {qiu2018netmf}
\bibfield{author}{\bibinfo{person}{Jiezhong Qiu}, \bibinfo{person}{Yuxiao
  Dong}, \bibinfo{person}{Hao Ma}, \bibinfo{person}{Jian Li},
  \bibinfo{person}{Kuansan Wang}, {and} \bibinfo{person}{Jie Tang}.}
  \bibinfo{year}{2018}\natexlab{}.
\newblock \showarticletitle{Network embedding as matrix factorization: Unifying
  deepwalk, line, pte, and node2vec}. In \bibinfo{booktitle}{\emph{Proceedings
  of the eleventh ACM international conference on web search and data mining}}.
  \bibinfo{pages}{459--467}.
\newblock


\bibitem[Shi et~al\mbox{.}(2018)]%
        {shi2018heterogeneous}
\bibfield{author}{\bibinfo{person}{Chuan Shi}, \bibinfo{person}{Binbin Hu},
  \bibinfo{person}{Wayne~Xin Zhao}, {and} \bibinfo{person}{S~Yu Philip}.}
  \bibinfo{year}{2018}\natexlab{}.
\newblock \showarticletitle{Heterogeneous information network embedding for
  recommendation}.
\newblock \bibinfo{journal}{\emph{IEEE Transactions on Knowledge and Data
  Engineering}} \bibinfo{volume}{31}, \bibinfo{number}{2}
  (\bibinfo{year}{2018}), \bibinfo{pages}{357--370}.
\newblock


\bibitem[Singer et~al\mbox{.}(2019)]%
        {singer2019tne}
\bibfield{author}{\bibinfo{person}{Uriel Singer}, \bibinfo{person}{Ido Guy},
  {and} \bibinfo{person}{Kira Radinsky}.} \bibinfo{year}{2019}\natexlab{}.
\newblock \showarticletitle{Node embedding over temporal graphs}. In
  \bibinfo{booktitle}{\emph{Proceedings of the 28th International Joint
  Conference on Artificial Intelligence}}. \bibinfo{pages}{4605--4612}.
\newblock


\bibitem[Su et~al\mbox{.}(2020)]%
        {su2020network}
\bibfield{author}{\bibinfo{person}{Chang Su}, \bibinfo{person}{Jie Tong},
  \bibinfo{person}{Yongjun Zhu}, \bibinfo{person}{Peng Cui}, {and}
  \bibinfo{person}{Fei Wang}.} \bibinfo{year}{2020}\natexlab{}.
\newblock \showarticletitle{Network embedding in biomedical data science}.
\newblock \bibinfo{journal}{\emph{Briefings in bioinformatics}}
  \bibinfo{volume}{21}, \bibinfo{number}{1} (\bibinfo{year}{2020}),
  \bibinfo{pages}{182--197}.
\newblock


\bibitem[Tang et~al\mbox{.}(2015)]%
        {tang2015line}
\bibfield{author}{\bibinfo{person}{Jian Tang}, \bibinfo{person}{Meng Qu},
  \bibinfo{person}{Mingzhe Wang}, \bibinfo{person}{Ming Zhang},
  \bibinfo{person}{Jun Yan}, {and} \bibinfo{person}{Qiaozhu Mei}.}
  \bibinfo{year}{2015}\natexlab{}.
\newblock \showarticletitle{Line: Large-scale information network embedding}.
  In \bibinfo{booktitle}{\emph{Proceedings of the 24th international conference
  on world wide web}}. \bibinfo{pages}{1067--1077}.
\newblock


\bibitem[Tang and Liu(2009a)]%
        {flickr}
\bibfield{author}{\bibinfo{person}{Lei Tang} {and} \bibinfo{person}{Huan Liu}.}
  \bibinfo{year}{2009}\natexlab{a}.
\newblock \showarticletitle{Relational learning via latent social dimensions}.
  In \bibinfo{booktitle}{\emph{Proceedings of the 15th ACM SIGKDD international
  conference on Knowledge discovery and data mining}}.
  \bibinfo{pages}{817--826}.
\newblock


\bibitem[Tang and Liu(2009b)]%
        {youtube}
\bibfield{author}{\bibinfo{person}{Lei Tang} {and} \bibinfo{person}{Huan Liu}.}
  \bibinfo{year}{2009}\natexlab{b}.
\newblock \showarticletitle{Scalable learning of collective behavior based on
  sparse social dimensions}. In \bibinfo{booktitle}{\emph{Proceedings of the
  18th ACM conference on Information and knowledge management}}.
  \bibinfo{pages}{1107--1116}.
\newblock


\bibitem[Tsitsulin et~al\mbox{.}(2018)]%
        {tsitsulin2018verse}
\bibfield{author}{\bibinfo{person}{Anton Tsitsulin}, \bibinfo{person}{Davide
  Mottin}, \bibinfo{person}{Panagiotis Karras}, {and} \bibinfo{person}{Emmanuel
  M{\"u}ller}.} \bibinfo{year}{2018}\natexlab{}.
\newblock \showarticletitle{Verse: Versatile graph embeddings from similarity
  measures}. In \bibinfo{booktitle}{\emph{Proceedings of the 2018 world wide
  web conference}}. \bibinfo{pages}{539--548}.
\newblock


\bibitem[Wang et~al\mbox{.}(2021)]%
        {wang2021approximate}
\bibfield{author}{\bibinfo{person}{Hanzhi Wang}, \bibinfo{person}{Mingguo He},
  \bibinfo{person}{Zhewei Wei}, \bibinfo{person}{Sibo Wang},
  \bibinfo{person}{Ye Yuan}, \bibinfo{person}{Xiaoyong Du}, {and}
  \bibinfo{person}{Ji-Rong Wen}.} \bibinfo{year}{2021}\natexlab{}.
\newblock \showarticletitle{Approximate graph propagation}. In
  \bibinfo{booktitle}{\emph{Proceedings of the 27th ACM SIGKDD Conference on
  Knowledge Discovery \& Data Mining}}. \bibinfo{pages}{1686--1696}.
\newblock


\bibitem[Wang et~al\mbox{.}(2018)]%
        {wang2018graphgan}
\bibfield{author}{\bibinfo{person}{Hongwei Wang}, \bibinfo{person}{Jia Wang},
  \bibinfo{person}{Jialin Wang}, \bibinfo{person}{Miao Zhao},
  \bibinfo{person}{Weinan Zhang}, \bibinfo{person}{Fuzheng Zhang},
  \bibinfo{person}{Xing Xie}, {and} \bibinfo{person}{Minyi Guo}.}
  \bibinfo{year}{2018}\natexlab{}.
\newblock \showarticletitle{Graphgan: Graph representation learning with
  generative adversarial nets}. In \bibinfo{booktitle}{\emph{Proceedings of the
  AAAI conference on artificial intelligence}}, Vol.~\bibinfo{volume}{32}.
\newblock


\bibitem[Wen et~al\mbox{.}(2018)]%
        {wen2018network}
\bibfield{author}{\bibinfo{person}{Yufei Wen}, \bibinfo{person}{Lei Guo},
  \bibinfo{person}{Zhumin Chen}, {and} \bibinfo{person}{Jun Ma}.}
  \bibinfo{year}{2018}\natexlab{}.
\newblock \showarticletitle{Network embedding based recommendation method in
  social networks}. In \bibinfo{booktitle}{\emph{Companion Proceedings of the
  The Web Conference 2018}}. \bibinfo{pages}{11--12}.
\newblock


\bibitem[Yang et~al\mbox{.}(2020)]%
        {yang2020nrp}
\bibfield{author}{\bibinfo{person}{Renchi Yang}, \bibinfo{person}{Jieming Shi},
  \bibinfo{person}{Xiaokui Xiao}, \bibinfo{person}{Yin Yang}, {and}
  \bibinfo{person}{Sourav~S Bhowmick}.} \bibinfo{year}{2020}\natexlab{}.
\newblock \showarticletitle{Homogeneous network embedding for massive graphs
  via reweighted personalized PageRank}.
\newblock \bibinfo{journal}{\emph{Proceedings of the VLDB Endowment}}
  \bibinfo{volume}{13}, \bibinfo{number}{5} (\bibinfo{year}{2020}),
  \bibinfo{pages}{670--683}.
\newblock


\bibitem[Yin and Wei(2019a)]%
        {yin2019strap}
\bibfield{author}{\bibinfo{person}{Yuan Yin} {and} \bibinfo{person}{Zhewei
  Wei}.} \bibinfo{year}{2019}\natexlab{a}.
\newblock \showarticletitle{Scalable graph embeddings via sparse transpose
  proximities}. In \bibinfo{booktitle}{\emph{Proceedings of the 25th ACM SIGKDD
  International Conference on Knowledge Discovery \& Data Mining}}.
  \bibinfo{pages}{1429--1437}.
\newblock


\bibitem[Yin and Wei(2019b)]%
        {yin2019scalable}
\bibfield{author}{\bibinfo{person}{Yuan Yin} {and} \bibinfo{person}{Zhewei
  Wei}.} \bibinfo{year}{2019}\natexlab{b}.
\newblock \showarticletitle{Scalable graph embeddings via sparse transpose
  proximities}. In \bibinfo{booktitle}{\emph{Proceedings of the 25th ACM SIGKDD
  International Conference on Knowledge Discovery \& Data Mining}}.
  \bibinfo{pages}{1429--1437}.
\newblock


\bibitem[Yu et~al\mbox{.}(2018)]%
        {yu2018netwalk}
\bibfield{author}{\bibinfo{person}{Wenchao Yu}, \bibinfo{person}{Wei Cheng},
  \bibinfo{person}{Charu~C Aggarwal}, \bibinfo{person}{Kai Zhang},
  \bibinfo{person}{Haifeng Chen}, {and} \bibinfo{person}{Wei Wang}.}
  \bibinfo{year}{2018}\natexlab{}.
\newblock \showarticletitle{Netwalk: A flexible deep embedding approach for
  anomaly detection in dynamic networks}. In
  \bibinfo{booktitle}{\emph{Proceedings of the 24th ACM SIGKDD international
  conference on knowledge discovery \& data mining}}.
  \bibinfo{pages}{2672--2681}.
\newblock


\bibitem[Zha and Simon(1999)]%
        {zha1999updating}
\bibfield{author}{\bibinfo{person}{Hongyuan Zha} {and} \bibinfo{person}{Horst~D
  Simon}.} \bibinfo{year}{1999}\natexlab{}.
\newblock \showarticletitle{On updating problems in latent semantic indexing}.
\newblock \bibinfo{journal}{\emph{SIAM Journal on Scientific Computing}}
  \bibinfo{volume}{21}, \bibinfo{number}{2} (\bibinfo{year}{1999}),
  \bibinfo{pages}{782--791}.
\newblock


\bibitem[Zhang et~al\mbox{.}(2016)]%
        {zhang2016approximate}
\bibfield{author}{\bibinfo{person}{Hongyang Zhang}, \bibinfo{person}{Peter
  Lofgren}, {and} \bibinfo{person}{Ashish Goel}.}
  \bibinfo{year}{2016}\natexlab{}.
\newblock \showarticletitle{Approximate personalized pagerank on dynamic
  graphs}. In \bibinfo{booktitle}{\emph{Proceedings of the 22nd ACM SIGKDD
  international conference on knowledge discovery and data mining}}.
  \bibinfo{pages}{1315--1324}.
\newblock


\bibitem[Zhang et~al\mbox{.}(2019)]%
        {zhang2019prone}
\bibfield{author}{\bibinfo{person}{Jie Zhang}, \bibinfo{person}{Yuxiao Dong},
  \bibinfo{person}{Yan Wang}, \bibinfo{person}{Jie Tang}, {and}
  \bibinfo{person}{Ming Ding}.} \bibinfo{year}{2019}\natexlab{}.
\newblock \showarticletitle{ProNE: Fast and Scalable Network Representation
  Learning.}. In \bibinfo{booktitle}{\emph{IJCAI}}, Vol.~\bibinfo{volume}{19}.
  \bibinfo{pages}{4278--4284}.
\newblock


\bibitem[Zhang et~al\mbox{.}(2018a)]%
        {zhang2018billion}
\bibfield{author}{\bibinfo{person}{Ziwei Zhang}, \bibinfo{person}{Peng Cui},
  \bibinfo{person}{Haoyang Li}, \bibinfo{person}{Xiao Wang}, {and}
  \bibinfo{person}{Wenwu Zhu}.} \bibinfo{year}{2018}\natexlab{a}.
\newblock \showarticletitle{Billion-scale network embedding with iterative
  random projection}. In \bibinfo{booktitle}{\emph{2018 IEEE International
  Conference on Data Mining (ICDM)}}. IEEE, \bibinfo{pages}{787--796}.
\newblock


\bibitem[Zhang et~al\mbox{.}(2018b)]%
        {zhang2018timers}
\bibfield{author}{\bibinfo{person}{Ziwei Zhang}, \bibinfo{person}{Peng Cui},
  \bibinfo{person}{Jian Pei}, \bibinfo{person}{Xiao Wang}, {and}
  \bibinfo{person}{Wenwu Zhu}.} \bibinfo{year}{2018}\natexlab{b}.
\newblock \showarticletitle{Timers: Error-bounded svd restart on dynamic
  networks}. In \bibinfo{booktitle}{\emph{Proceedings of the AAAI Conference on
  Artificial Intelligence}}, Vol.~\bibinfo{volume}{32}.
\newblock


\bibitem[Zhang et~al\mbox{.}(2018c)]%
        {zhang2018arbitrary}
\bibfield{author}{\bibinfo{person}{Ziwei Zhang}, \bibinfo{person}{Peng Cui},
  \bibinfo{person}{Xiao Wang}, \bibinfo{person}{Jian Pei},
  \bibinfo{person}{Xuanrong Yao}, {and} \bibinfo{person}{Wenwu Zhu}.}
  \bibinfo{year}{2018}\natexlab{c}.
\newblock \showarticletitle{Arbitrary-order proximity preserved network
  embedding}. In \bibinfo{booktitle}{\emph{Proceedings of the 24th ACM SIGKDD
  international conference on knowledge discovery \& data mining}}.
  \bibinfo{pages}{2778--2786}.
\newblock


\bibitem[Zheng et~al\mbox{.}(2022)]%
        {zheng2022instant}
\bibfield{author}{\bibinfo{person}{Yanping Zheng}, \bibinfo{person}{Hanzhi
  Wang}, \bibinfo{person}{Zhewei Wei}, \bibinfo{person}{Jiajun Liu}, {and}
  \bibinfo{person}{Sibo Wang}.} \bibinfo{year}{2022}\natexlab{}.
\newblock \showarticletitle{Instant Graph Neural Networks for Dynamic Graphs}.
  In \bibinfo{booktitle}{\emph{Proceedings of the 28th ACM SIGKDD Conference on
  Knowledge Discovery and Data Mining}} \emph{(\bibinfo{series}{KDD '22})}.
  \bibinfo{pages}{2605–2615}.
\newblock
\showISBNx{9781450393850}


\bibitem[Zhou et~al\mbox{.}(2018)]%
        {zhou2018dynamictrad}
\bibfield{author}{\bibinfo{person}{Lekui Zhou}, \bibinfo{person}{Yang Yang},
  \bibinfo{person}{Xiang Ren}, \bibinfo{person}{Fei Wu}, {and}
  \bibinfo{person}{Yueting Zhuang}.} \bibinfo{year}{2018}\natexlab{}.
\newblock \showarticletitle{Dynamic network embedding by modeling triadic
  closure process}. In \bibinfo{booktitle}{\emph{Proceedings of the AAAI
  conference on artificial intelligence}}, Vol.~\bibinfo{volume}{32}.
\newblock


\bibitem[Zhu et~al\mbox{.}(2018)]%
        {zhu2018high}
\bibfield{author}{\bibinfo{person}{Dingyuan Zhu}, \bibinfo{person}{Peng Cui},
  \bibinfo{person}{Ziwei Zhang}, \bibinfo{person}{Jian Pei}, {and}
  \bibinfo{person}{Wenwu Zhu}.} \bibinfo{year}{2018}\natexlab{}.
\newblock \showarticletitle{High-order proximity preserved embedding for
  dynamic networks}.
\newblock \bibinfo{journal}{\emph{IEEE Transactions on Knowledge and Data
  Engineering}} \bibinfo{volume}{30}, \bibinfo{number}{11}
  (\bibinfo{year}{2018}), \bibinfo{pages}{2134--2144}.
\newblock


\end{thebibliography}
